\documentclass[10pt]{article}
\usepackage[dvipsnames]{xcolor}
\usepackage[accepted]{tmlr}
\usepackage{amsmath,amsfonts,bm}

\def\eqref#1{equation~\ref{#1}}
\def\1{\bm{1}}

\DeclareMathAlphabet{\mathsfit}{\encodingdefault}{\sfdefault}{m}{sl}
\SetMathAlphabet{\mathsfit}{bold}{\encodingdefault}{\sfdefault}{bx}{n}

\usepackage{hyperref}
\usepackage{url}

\usepackage{array,multirow,graphicx}

\usepackage{amsmath}
\usepackage{algorithm}
\usepackage{algpseudocode}
\usepackage{graphicx}
\usepackage{subcaption}
\usepackage{caption}
\usepackage{lipsum}  % For generating dummy text
\usepackage{mwe}
\usepackage{amssymb}
\usepackage{appendix}
\usepackage{enumitem}
\usepackage{ulem}
\usepackage[margin=1in]{geometry} % to set margins
\usepackage{booktabs} % for better borders
\usepackage{siunitx} % for aligning numbers by decimal point
\usepackage[margin=1in]{geometry} % to set margins
\definecolor{lightgray}{gray}{0.9}
\usepackage{xcolor}
\usepackage{tabularray}
\usepackage{amsthm}
\newtheorem{thm}{Theorem}[section]

\title{Scaling Up Bayesian Neural Networks with Neural Networks}

\author{\name Zahra Moslemi \email zmoslemi@uci.edu \\
       \addr Department of Statistics\\
       University of California\\
       Irvine, CA, USA
       \AND
       \name Yang Meng \email mengy13@uci.edu \\
       \addr Department of Statistics\\
       University of California\\
       Irvine, CA, USA
       \AND
       \name Shiwei Lan \email slan7@asu.edu \\
       \addr School of Mathematical and Statistical Sciences\\
       Arizona State University\\
       Tempe, AZ, USA
       \AND
       \name Babak Shahbaba \email babaks@uci.edu \\
       \addr Department of Statistics\\
       University of California\\
       Irvine, CA, USA}

\def\month{09}  % Insert correct month for camera-ready version
\def\year{2024} % Insert correct year for camera-ready version
\def\openreview{\url{https://openreview.net/forum?id=cD209UgOX7}} % Insert correct link to OpenReview for camera-ready version

\begin{document}

\maketitle

\begin{abstract}
Bayesian Neural Networks (BNNs) offer a principled and natural framework for proper uncertainty quantification in the context of deep learning. They address the typical challenges associated with conventional deep learning methods, such as data insatiability, ad-hoc nature, and susceptibility to overfitting. However, their implementation typically either relies on Markov chain Monte Carlo (MCMC) methods, which are characterized by their computational intensity and inefficiency in a high-dimensional space, or variational inference methods, which tend to underestimate uncertainty. To address this issue, we propose a novel Calibration-Emulation-Sampling (CES) strategy to significantly enhance the computational efficiency of BNN. In this framework, during the initial calibration stage, we collect a small set of samples from the parameter space. These samples serve as training data for the emulator, which approximates the map between parameters and posterior probability. The trained emulator is then used for sampling from the posterior distribution at substantially higher speed compared to the standard BNN. Using simulated and real data, we demonstrate that our proposed method improves computational efficiency of BNN, while maintaining similar performance in terms of prediction accuracy and uncertainty quantification. 
\end{abstract}

\section{Introduction}

In recent years, Deep Neural Networks (DNN) have emerged as the predominant driving force in the field of machine learning and are regarded as fundamental tools for many intelligent systems \citep{cheng2018,lecun1998,sze2017}. While DNN have demonstrated significant success in prediction tasks, they often struggle with accurately quantifying uncertainty. Additionally, due to their vulnerability to overfitting, they can generate highly confident yet erroneous predictions \citep{su2019, kwon2022}. In recent years, there have been some attempts to address this issue. For example, the Ensemble Deep Learning method \citep{lakshminarayanan2017} aggregates predictions from multiple models to improve reliability and uncertainty estimates. While these methods represent important progress in the right direction, developing a principled and computationally efficient framework for Uncertainty Quantification (UQ) within the context of deep learning remains a significant challenge and active area of research. This is especially important in domains
where critical decisions, such as medical diagnostics, are involved. To address these issues, Bayesian Neural Networks (BNNs) \citep{mackay1992, neal2012bayesian, jospin2022} have emerged as an alternative to standard DNN, providing a more reliable framework within the field of machine learning. Their intrinsic ability to capture and quantify uncertainties in predictions establishes a robust foundation for decision-making under uncertainty. However, Bayesian inference in high-dimensional BNN poses significant computational challenges due to the inefficiency of traditional Markov Chain Monte Carlo (MCMC) methods. In fact, not only BNN, but almost all traditional Bayesian inference methods relying on MCMC techniques are known for their computational intensity and inefficiency when dealing with high-dimensional problems. Variational inference \citep{jordan1999introduction} methods have been proposed to speed up computation by approximating the posterior distribution, but they tend to underestimate uncertainty \citep{minka2005}. Consequently, researchers have proposed various approaches to expedite the inference process \citep{welling2011bayesian, shahbabaSplitHMC, AhnShahbabaWelling14, Hoffman14a, beskos2017, cui2016, Zhang2017b, Zhang2018, li2019neural}. Here, we focus on a state-of-the-art approach, called Calibration-Emulation-Sampling (CES) \citep{Cleary2021}, which has shown promising results in large-dimensional UQ problems such as inverse problems \citep{Lan2022Siam}. CES involves the following three steps:
\begin{description}[itemsep=0pt,parsep=6pt]
\item[(i)] Calibrate models to collect sample parameters and their corresponding expensive evaluation of posterior probability for the emulation step; 
\item[(ii)] Emulate the parameter-to-posterior map using the samples from Step (i); and 
\item[(iii)] Generate posterior samples using MCMC based on the trained emulator at substantially lower cost. 
\end{description}

This framework allows for reusing expensive forward evaluations from parameters to posterior probability and offers a computationally efficient alternative to existing MCMC procedures.

The standard CES method \citep{Cleary2021} focuses on UQ in inverse problems and uses Gaussian Process (GP) models for the emulation component. GP models have a well-established history of application in emulating computer models \citep{currin1988}, conducting uncertainty analyses \citep{Oakley2002}, sensitivity assessments \citep{Jermey2023}, and calibrating computer codes \citep{Kennedy2002,Higdon2004}. Despite their versatility, GP-based emulators are computationally intensive, with a complexity of O($N^3$) using the squared-exponential kernel, where $N$ is the sample size. Lower computational complexity can be achieved using alternative kernels \citep{lan15} or various computational techniques \citep{liu2020,bonilla2007,gardner2018,seeger2003}. Nevertheless, scaling up GP emulators to high-dimensional problems remains a limiting factor. Furthermore, the prediction accuracy of GP emulators highly depends on the quality of the training data, emphasizing the importance of rigorous experimental design. 
To address these issues, \citet{Lan2022Siam} proposed an alternative CES scheme called Dimension-Reduced Emulative Autoencoder Monte Carlo (DREAMC) method, which uses Convolutional Neural Networks (CNN) as emulator. DREAMC improves and scales up the application of the CES framework for Bayesian UQ in inverse problems from hundreds of dimensions (with GP emulation) to thousands of dimensions (with CNN emulation). Here, we adopt a similar approach and propose a new method, called Fast BNN (FBNN), for Bayesian inference in neural networks. We use DNN for the emulation component of our CES scheme. DNN has proven to be a powerful tool in a variety of applications and offers several advantages over GP emulation \citep{Lan2022Siam, dargan2020}. It is computationally more efficient and suitable for high-dimensional problems. The choice of DNN as an emulator enhances computational efficiency and flexibility.

Besides the computational challenges of building emulators, efficiently sampling from posterior distributions using these emulators also presents a significant challenge due to the high dimensionality of the target distribution. Traditional Metropolis-Hastings algorithms, typically defined on finite-dimensional spaces, encounter diminishing mixing efficiency as the dimensions increase \citep{gelman1997, roberts1998, beskos2009}. To overcome this inherent drawback, a novel class of {dimension-independent} MCMC methods has emerged, operating within infinite-dimensional spaces \citep{beskos2014,beskos2009,beskos2011,cotter2013,law2014,beskos2014,beskos2017}. More specifically, we use the Preconditioned Crank-Nicolson (pCN) algorithm. The most significant feature of pCN is its dimension robustness, which makes it well-suited for high-dimensional sampling problems. The pCN algorithm is well-defined, with non-degenerate acceptance probability, even for target distributions on infinite-dimensional spaces. As a result, when pCN is implemented on a real-world computer in large but finite dimension $N$, the convergence properties of the algorithm are independent of $N$. This is in strong contrast to schemes such as Gaussian random-walk Metropolis-Hastings and the Metropolis-adjusted Langevin algorithm, whose acceptance probability degenerates to zero as $N$ goes to infinity. 

In summary, this paper addresses the critical challenges of UQ in high-dimensional BNN. By incorporating deep neural networks for emulation and leveraging the dimension-robust pCN algorithm for sampling, this research significantly enhances computational efficiency and scalability in Bayesian uncertainty quantification, offering a robust counterpart to DNN, and a scalable counterpart to BNN.
Through extensive experiments, we demonstrate the feasibility and effectiveness of utilizing FBNN to accelerate Bayesian UQ in high-dimensional neural networks. %The proposed method showcases remarkable computational efficiency, enabling scalable Bayesian inference in BNN with thousands of dimensions.

\section{Related Methods}
\label{sec:methods}
Various MCMC methods have been employed to explore complex probability distributions for Bayesian inference. In this section, we discuss some of the main MCMC algorithms related to our work. Additionally, we discuss a variety of state-of-the-art methods utilized in our numerical experiments, which extend beyond MCMC frameworks. These include Ensemble Deep Learning for Neural Networks \citep{perrone1995networks}, BNNs with Variational Inference \citep{jaakkola2000bayesian}, BNNs leveraging Lasso Approximation \citep{mackay1992}, Monte Carlo Dropout (MC-Dropout) \citep{gal2016dropout}, Stochastic Weight Averaging-Gaussian (SWAG) \citep{maddox2019simple}, and Accelerated Hamiltonian Monte Carlo (HMC) \citep{Zhang2017b}. These techniques offer a comprehensive spectrum for evaluating our FBNN model.

% \subsection{Metropolis-Hastings algorithm}

% Metropolis-Hastings is a fundamental MCMC method used for obtaining a sequence of random samples from a probability distribution for which direct sampling is difficult \citep{chib1995,robert1999}.
% The following algorithm outlines the step-by-step process for implementing the Metropolis-Hastings method.

% \begin{algorithm}[t]
%   \caption{Metropolis-Hastings Algorithm}
%   \begin{algorithmic}[1]
%     \State \textbf{Initialization:}
%     \State Choose an initial state, $X_0$.
%     \State Set the iteration counter $t = 0$.

%     \State

%     \While{not enough samples obtained}
%       \State \textbf{Proposal:}
%       \State Generate a proposal state $Y$ from a proposal distribution $q(Y | X_t)$.

%       \State

%       \State \textbf{Acceptance Probability:}
%       \State Compute the acceptance probability, $\alpha$:
%       \[
%         \alpha = \min\left(1, \frac{p(Y) \cdot q(X_t | Y)}{p(X_t) \cdot q(Y | X_t)}\right)
%       \]
%       where:
%       \begin{itemize}
%         \item $p(\cdot)$ is the target distribution.
%         \item $q(\cdot | \cdot)$ is the proposal distribution.
%       \end{itemize}

%       \State

%       \State \textbf{Accept/Reject:}
%       \State Generate a uniform random number $u$ from the interval $[0, 1]$.
%       \If{$u \leq \alpha$}
%         \State Set $X_{t+1} = Y$.
%       \Else
%         \State Set $X_{t+1} = X_t$.
%       \EndIf

%       \State

%       \State \textbf{Update:}
%       \State Increment the iteration counter: $t = t + 1$.

%     \EndWhile

%   \end{algorithmic}
% \end{algorithm}

\subsection{Hamiltonian Monte Carlo (HMC)}
MCMC methods are designed for sampling from intractable probability distributions. The fundamental principle involves constructing a Markov chain whose equilibrium distribution coincides with the target distribution. Various algorithms exist for constructing such Markov chains, the simplest of which is the Metropolis-Hastings (MH) algorithm \citep{metropolis1953equation, hastings70}, which relies on a random walk to explore the parameter space.
Hamiltonian Monte Carlo (HMC) is a special case of the MH algorithm that incorporates Hamiltonian dynamics evolution and auxiliary momentum variables \citep{neal2011mcmc}. Compared to using a Gaussian random walk proposal distribution in the MH algorithm, HMC reduces the correlation between successive sampled states by proposing moves to distant states that maintain a high probability of acceptance due to the approximate energy conserving properties of the simulated Hamiltonian dynamic. The reduced correlation means fewer Markov chain samples are needed to approximate integrals with respect to the target probability distribution for a given Monte Carlo error.

\subsection{Stochastic Gradient Hamiltonian Monte Carlo (SGHMC)}

As discussed earlier, HMC sampling methods provide a mechanism for defining distant proposals with high acceptance probabilities in a Metropolis-Hastings framework, enabling more efficient exploration of the state space than standard random-walk proposals. However, a limitation of HMC methods is the required gradient computation for simulation of the Hamiltonian dynamical system; such computation is infeasible in problems involving a large sample size. Stochastic Gradient Hamiltonian Monte Carlo (SGHMC) \citep{chen2014} addresses computational inefficiency by using a noisy but unbiased estimate of the gradient computed from a mini-batch of the data. SGHMC is an effective method for Bayesian inference, particularly when dealing with large datasets, as it leverages stochastic gradients and hyperparameter adaptation to efficiently explore high-dimensional target distributions. %, presented in algorithm \ref{sghmc}.

\subsection{Random Network Surrogate–HMC (RNS–HMC)}

Alternatively, we can reduce the computational cost of HMC by constructing surrogate Hamiltonians. As an example, the random network surrogate–HMC (RNS–HMC) method \citep{Zhang2017b} uses random non-linear bases to approximate posterior distributions. The goal is to explore and exploit the structure and regularity in parameter space for the underlying probabilistic model and construct an effective approximation of its geometric properties. To achieve this, RNS–HMC starts by identifying suitable bases that can capture the complex geometric properties of the parameter space. Through an optimization process, these bases are then used to form a surrogate function to approximate the expensive Hamiltonian dynamics. Unlike traditional HMC, which requires repeated evaluation of the model and its derivatives, RNS-HMC leverages the surrogate function to perform leapfrog integration steps, leading to a substantially lower computational cost. Later, \citet{li2019neural} extended this idea by using a neural network to directly approximate the gradient of the Hamiltonian.

\subsection{Preconditioned Crank-Nicolson (pCN)}
\label{subsec:pcn}

Preconditioned Crank-Nicolson \citep{da2014} is a variant of MCMC that incorporates a preconditioning matrix for adaptive scaling \citep{cotter2013}.
It involves re-scaling and random perturbation of the current state, incorporating prior information. Despite the Gaussian prior assumption, the approach adapts to cases where the posterior distribution may not be Gaussian but is absolutely continuous with respect to an appropriate Gaussian density. This adaptation is achieved through the Radon-Nikodym derivative, connecting the posterior distribution with the dominating Gaussian measure, often chosen as the prior. The algorithmic foundation of pCN lies in using stochastic processes that preserve either the posterior or prior distribution. These processes serve as proposals for Metropolis-Hastings methods with specific discretizations, ensuring preservation of the Gaussian reference measure. The key steps of the pCN algorithm are outlined in Algorithm \ref{alg:pcn}.

\begin{algorithm}[t]
\caption{Preconditioned Crank-Nicolson (pCN) Algorithm}
\label{alg:pcn}
\begin{algorithmic}
    \State \textbf{Notation:}
    \State \hspace{\algorithmicindent} $k$: iteration index.
    \State \hspace{\algorithmicindent} $u^{(k)}$: state of the algorithm at iteration $k$.
    \State \hspace{\algorithmicindent} $\xi^{(k)} \sim \mathcal{N}(0,C)$: Gaussian noise with covariance matrix $C$ at iteration $k$.
    \State \hspace{\algorithmicindent} $\beta$: parameter controlling the proposal step size.
    \State \hspace{\algorithmicindent} $\Phi$: potential function related to the target probability distribution.
    \State \hspace{\algorithmicindent} $a(u^{(k)}, v^{(k)})$: acceptance probability for the proposed state $v^{(k)}$.
    \State
    \State \textbf{Algorithm Steps:}
    \State Initialize iteration counter $k \gets 0$.
    \State Choose an initial state $u^{(0)}$.
    \Repeat
        \State Generate noise $\xi^{(k)}$ with $\xi^{(k)} \sim \mathcal{N}(0,C)$.
        \State Propose $v^{(k)} \gets \sqrt{1-\beta^2}u^{(k)} + \beta \xi^{(k)}$.
        \State Calculate acceptance probability $a(u^{(k)}, v^{(k)}) \gets \min\{1, \exp(\Phi(u^{(k)}) - \Phi(v^{(k)}))\}$.
        \If{random number $\leq a(u^{(k)}, v^{(k)})$}
            \State Set $u^{(k+1)} \gets v^{(k)}$. 
        \Else
            \State Set $u^{(k+1)} \gets u^{(k)}$.
        \EndIf
        \State Increment iteration counter $k \gets k + 1$.
    \Until{a stopping criterion (fixed number of iterations or convergence) is met.}
\end{algorithmic}
\end{algorithm}

\subsection{Variational Inference}

The concept of variational inference has been applied in various forms to probabilistic models. The technique offers a way to approximate posterior distributions in Bayesian models \citep{jordan1999introduction}. The approximate distribution allows for a more feasible inference, especially for complex models like neural networks. In the context of BNNs, variational inference was brought into focus by \citet{hinton1993keeping}, which, while not explicitly termed as such in the modern sense, laid the groundwork for later developments. A more direct application of variational inference to BNNs was detailed in later works \citep[e.g., ][]{graves2011practical}. More recently, \citet{Kingma2014VAE}, \citet{rezende2014stochastic} and \citet{blundell2015weight} significantly contributed to popularizing and advancing the use of variational inference in deep learning and Bayesian neural networks through the introduction of efficient gradient-based optimization techniques.

Algorithm \ref{alg:variational_inference_BNN} succinctly shows the iterative process of optimizing the parameters of a variational distribution to approximate the posterior distribution of a BNN's weights. Through the alternation of expectation (E-Step) and maximization (M-Step) phases, it aims to minimize the difference between the variational distribution and the true posterior, leveraging the Evidence Lower Bound (ELBO) \citep{jordan1999introduction} as a tractable surrogate objective function. This approach enables the practical application of Bayesian inference to neural networks, facilitating the quantification of uncertainty in predictions and model parameters. %See Algorithm \ref{alg:variational_inference_BNN} for more details. 

\begin{algorithm}[t]
  \caption{Variational Inference in Bayesian Neural Networks (BNNs)}
  \label{alg:variational_inference_BNN}
  \begin{algorithmic}
    \State \textbf{Initialization:}
    \State Choose an initial variational distribution $q_{\theta}(W)$ for the weights $W$ of BNN, parameterized by $\theta$.
    \State Define the prior distribution $p(W)$ over the weights.

    \State

    \While{not converged}
      \State \textbf{E-Step:} Estimate the Expectation of the log-likelihood over the variational distribution.
      \State Compute the gradient of the ELBO (Evidence Lower BOund) with respect to $\theta$, where
      \[
        \text{ELBO}(\theta) = \mathbb{E}_{q_{\theta}(W)}[\log p(Y|X,W)] - \text{KL}[q_{\theta}(W) || p(W)]
      \]
      \State Here, $X$ and $Y$ are the inputs and outputs of the dataset, respectively, and $\text{KL}$ denotes the Kullback-Leibler divergence between the variational distribution and the prior.

      \State

      \State \textbf{M-Step:} Maximize the ELBO with respect to $\theta$ using gradient ascent:
      \[
        \theta \leftarrow \theta + \eta \nabla_{\theta} \text{ELBO}(\theta)
      \]
      \State where $\eta$ is the learning rate.

    \EndWhile

    \State

    \State \textbf{Output:} Variational distribution $q_{\theta}(W)$ approximating the posterior distribution $p(W|X,Y)$.

  \end{algorithmic}
\end{algorithm}

\subsection{Laplace Approximation}

Previous studies have shown that in the context of BNNs, the Laplace approximation serves as an efficient method for approximating the posterior distribution over the network's weights \citep{arbel2023primer, blundell2015weight}. At the core of the Laplace approximation is the assumption that, around the loss function's minimum, the posterior distribution of the network's weights can be approximated by a Gaussian distribution. This is achieved by finding the mode of the posterior, called the Maximum A Posteriori (MAP; equivalent to the minimum of the loss function in the Bayesian framework), and then approximating the curvature of the loss surface at this point using the Hessian matrix \citep{liang2018bayesian}. The inverse of this Hessian is used to define the covariance of the Gaussian posterior, thus simplifying the representation of uncertainty in the predictions. More specifically in BNNs, this approach can be used to approximate the posterior distribution of the weights given the observed data.

\subsection{Monte Carlo Dropout}
Monte Carlo (MC) Dropout \citep{gal2016dropout} was introduced as a Bayesian approximation method to quantify model uncertainty in deep learning. The core idea behind this method is to interpret dropout, a technique commonly used to prevent overfitting in neural networks, from a Bayesian perspective. Normally, dropout randomly disables a fraction of neurons during the training phase to improve generalization. However, when viewed through the Bayesian lens, dropout can be seen as a practical way to approximate Bayesian inference in deep networks. This approximation allows the network to estimate not just a single set of weights, but a distribution over them, enabling the model to express uncertainty in its predictions. The MC Dropout technique involves running multiple forward passes through the network with dropout enabled. Each forward pass generates a different set of predictions due to the random omission of neurons, leading to a distribution of outputs for a given input.

\subsection{Stochastic Weight Averaging-Gaussian (SWAG)}
Building on the idea of Stochastic Weight Averaging (SWA) \citep{izmailov2018, maddox2019simple}, SWAG approximates the distribution of model weights by a Gaussian distribution, leveraging the empirical weight samples collected from training. This approach allows for a more nuanced understanding of the model's uncertainty compared to SWA, which simply averages weights over the latter part of the training process.
SWAG involves collecting a set of weights $\{W_i\}_{i=1}^N$ over the last $N$ epochs of training, where $W_i$ represents the weight vector at epoch $i$. The mean $\mu$ of the Gaussian distribution is then computed as the simple average of these weights: $\mu = \frac{1}{N} \sum_{i=1}^N W_i$. To capture the covariance of the weight distribution, SWAG calculates the empirical covariance matrix: $\Sigma = \frac{1}{N-1} \sum_{i=1}^N (W_i - \mu) (W_i - \mu)^T$. This formulation assumes a diagonal, or low-rank plus diagonal, approximation of the covariance matrix to maintain computational efficiency. The resulting Gaussian distribution, characterized by $\mu$ and $\Sigma$, can then be used for uncertainty estimation and prediction by sampling weights from this distribution and averaging the predictions of the resulting models.

\section{Bayesian UQ for Neural Networks: Calibration-Emulation-Sampling}
Standard neural networks (NN) typically consist of multiple layers, starting with an input layer, denoted as $\boldsymbol{l}_0$, followed by a series of hidden layers $\boldsymbol{l}_l$ for $l = 1, \ldots, m-1$, and ending with an output layer $\boldsymbol{l}_m$. In this architectural framework, comprising a total of $m+1$ layers, each layer $\boldsymbol{l}$ is characterized by a linear transformation, which is subsequently subjected to a nonlinear operation $g$, commonly referred to as an activation function \citep{jospin2022}:
\begin{eqnarray}\label{eq:NN-structure}
\boldsymbol{l}_0 & = & \boldsymbol{X}, \nonumber \\ 
\boldsymbol{l}_l & = & g_l\left(\boldsymbol{W}_l \boldsymbol{l}_{l-1} + \boldsymbol{b}_l\right) \quad \text{for all } l \in \{1,\cdots, m-1\}, \\
\boldsymbol{l}_m & = & \boldsymbol{Y}. \nonumber
\end{eqnarray}
% \begin{equation*}\label{eq:NN-structure}
% \boldsymbol{l}_0 = \boldsymbol{X}, 
% \end{equation*}
% \begin{equation}\label{eq:NN-structure1}
% \boldsymbol{l}_l = g_l\left(\boldsymbol{W}_l \boldsymbol{l}_{l-1} + \boldsymbol{b}_l\right) \quad \text{for all } l \in \{1,\cdots, m\},
% \end{equation}
% \begin{equation*}
% \boldsymbol{Y} = \boldsymbol{l}_m .
% \end{equation*}
Here, $\boldsymbol{\theta}=(\boldsymbol{W}, \boldsymbol{b})$ are the parameters of the network, where $\boldsymbol{W}$ are the weights of the network connections and $\boldsymbol{b}$ are the biases. A given NN architecture represents a set of functions isomorphic to the set of possible parameters $\boldsymbol{\theta}$. Deep learning is the process of estimating the parameters $\boldsymbol{\theta}$ from the training set $\boldsymbol{(X, Y):=\left\{\left(x_{n}, y_{n}\right)\right\}_{n=1}^{N}}$ composed of a series of input $\boldsymbol{X}$ and their corresponding labels $\boldsymbol{Y}$. Based on the training set, a neural network is trained to optimize network parameters $\boldsymbol{\theta}$ in order to map $\boldsymbol{X} \rightarrow \boldsymbol{Y}$ with the objective of obtaining the maximal accuracy (under certain loss function $L(\cdot))$. Considering the error, we can write NN as a forward mapping, denoted as $\mathcal{G}$, that maps each parameter vector $\boldsymbol{\theta}$ to a function that further connects $\boldsymbol{X}$ to $\boldsymbol{Y}$ with small errors $\varepsilon_{n}$ :
\begin{eqnarray}
\mathcal{G}: \Theta \rightarrow \boldsymbol{Y}^{\boldsymbol{X}}, \quad \boldsymbol{\theta} \mapsto \mathcal{G}(\boldsymbol{\theta})
\end{eqnarray}
More specifically,
\begin{eqnarray}
\mathcal{G}(\boldsymbol{\theta}): \boldsymbol{X} \rightarrow \boldsymbol{Y}, \quad \boldsymbol{y_{n}} = \boldsymbol{\hat{y}_{n}} + \boldsymbol{\varepsilon_{n}}, \quad \boldsymbol{\hat{y}_{n}} = \mathcal{G}(\boldsymbol{x_n}; \theta), \quad \boldsymbol{\varepsilon} \sim \boldsymbol{N(0, \Gamma)}
\end{eqnarray}
where $\varepsilon$ represents random noise capturing disparity between the predicted and actual observed values in the training data. Here, $\boldsymbol{Y}$ is a continuous random variable in regression problems, or a continuous {\it{latent}} variable in classification problems.

To train NN, stochastic gradient algorithms could be used to solve the following optimization problem:
\begin{eqnarray}
\boldsymbol{\theta}^{*} &=& \arg \min_{\boldsymbol{\theta} \in \Theta} L(\boldsymbol{\theta}; \boldsymbol{X}, \boldsymbol{Y}) = \arg \min_{\boldsymbol{\theta} \in \Theta} L(\boldsymbol{Y} - \mathcal{G}(\boldsymbol{X}; \theta)) \nonumber
\end{eqnarray}
For example, the loss function $L(\boldsymbol{\theta}; \boldsymbol{X}, \boldsymbol{Y})$ can be defined in terms of the negative log-likelihood function $\Phi$ as follows:
\begin{eqnarray}
\Phi(\boldsymbol{\theta}; \boldsymbol{X}, \boldsymbol{Y}) &=& \frac{1}{2}\|\boldsymbol{Y} - \mathcal{G}(\boldsymbol{X};\boldsymbol{\theta})\|_{\Gamma}^2
\end{eqnarray}

The point estimate approach, which is the traditional approach in deep learning, is relatively straightforward to implement with modern algorithms and software packages, but tends to lack proper uncertainty quantification \citep{chuan2017,nixon2019}. To address this issue, stochastic neural networks, which incorporate stochastic components in the network, have emerged as a standard solution. This is performed by giving the network either stochastic activation functions or stochastic weights to simulate random samples for $\boldsymbol{\theta}$. The integration of stochastic components into neural networks allows for an extensive exploration of model uncertainty, which can be approached through Bayesian methods among others. It should be noted that not all neural networks that represent uncertainty are Bayesian or even stochastic; some employ deterministic methods to estimate uncertainty without relying on stochastic components or Bayesian inference \citep{lakshminarayanan2017,sensoy2018Zhang2017b}. Bayesian neural networks (BNN) represent a subset of stochastic neural networks where Bayesian inference is specifically used for training, offering a rigorous probabilistic interpretation of model parameters \citep{mackay1992, neal2012bayesian}. The primary objective is to gain a deeper understanding of the uncertainty that underlies the specific process the network is modeling. 

To design a BNN, we put a prior distribution over the model parameters, $p(\boldsymbol{\theta})$. By applying Bayes' theorem, the posterior probability can be written as:
\begin{eqnarray}
p(\boldsymbol{\theta} \mid X,Y) & = & \frac{p\left(Y \mid X, \boldsymbol{\theta}\right) p(\boldsymbol{\theta})}{\int_{\boldsymbol{\theta}} p\left(Y \mid X, \boldsymbol{\theta}^{\prime}\right) p\left(\boldsymbol{\theta}^{\prime}\right) d \boldsymbol{\theta}^{\prime}} \\ 
& \propto & p\left(Y \mid X, \boldsymbol{\theta}\right) p(\boldsymbol{\theta}).
\end{eqnarray}

BNN is usually trained using MCMC algorithms. Because we typically have big amount of data, the likelihood evaluation tends to be expensive. One common approach to address this issue is subsampling, which restricts the computation to a subset of the data \citep[see for example,][]{hoffmann10, welling2011bayesian, chen2014}. The assumption is that there is redundancy in the data and an appropriate subset of the data can provide a good enough approximation of the information provided by the full data set. In practice, it is a challenge to find good criteria and strategies for an effective subsampling in many applications. Additionally, subsampling could lead to a significant loss of accuracy \citep{betancourt2015}.

\subsection{Fast Bayesian Neural Network (FBNN).}

\begin{algorithm}[t!]
\caption{Fast Bayesian Neural Network (FBNN)}\label{alg:fbnn}
\begin{algorithmic}
\State \textbf{Input:} Training set $\{(X_n, Y_n)\}_{n=1}^N$, Prior $p(\boldsymbol{\theta})$
\State \textbf{Output:} Posterior samples for model parameters
\Procedure{FBNN}{$\{(X_n, Y_n)\}_{n=1}^N$, $p(\boldsymbol{\theta})$}
    \State \textbf{Calibration Step:}
    \State Initialize model parameters $\boldsymbol{\theta}$ using SGHMC
    \State Save posterior samples $\{\boldsymbol{\theta}_n^{(j)}\}_{j=1}^J$ and the corresponding 
    $\{ \mathcal{G}_{\boldsymbol{\theta}_n^{(j)}}(\boldsymbol{X_{n}})\}_{j=1}^{J}$ after a few iterations
    
    \State \textbf{Emulation Step:}
    \State Build an emulator of the forward mapping $\mathcal{G}^{e}$ based on $\{\boldsymbol{\theta}_n^{(j)},\mathcal{G}_{\boldsymbol{\theta}_n^{(j)}}(\boldsymbol{X_{n}})\}_{j=1}^{J}$ using a DNN as the emulator
    
    \State \textbf{Sampling Step:}
    Run approximate MCMC, particularly $pCN$, based on the emulator to propose $\theta^\prime$ from $\theta$.

\EndProcedure
\end{algorithmic}
\end{algorithm}

We propose an alternative approach that explores smoothness or regularity in parameter space, a characteristic common to most statistical models. Therefore, one would expect to find good and compact forms of approximation of functions (e.g., likelihood function) in parameter space. Sampling algorithms can use these approximate functions, also known as ``surrogate'' functions, to reduce their computational cost. More specifically, we propose using the CES scheme for high-dimensional BNN problems in order to bypass the expensive evaluation of original forward models and reduce the cost of sampling to a small computational overhead. 
Compared with MCMC methods, which require repeatedly evaluating the original (large) NN for the likelihood given the data, the proposed method builds a smaller NN emulator that bypasses the data (i.e., cuts out the middleman) by mapping the parameters directly to the likelihood function, thus avoiding costly evaluations. That is, the emulator is trained based on the parameter-likelihood pairs, which are collected through few iterations of the original BNN. In contrast to subsampling methods, this approach can handle computationally intensive likelihood functions, whether the computational cost is due to high-dimensional data or complex likelihood function (e.g., models based on differential equations). 
Additionally, the calibration process increases the efficiency of MCMC algorithms by providing a robust initial point in the high-density region.
Algorithm \ref{alg:fbnn} shows how our proposed method, called Fast Bayesian Neural Networ (FBNN), combines the strengths of BNN in uncertainty quantification, SGHMC for efficient parameter calibration, and the pCN method for sampling. More details are provided in the following sections. 

%\textcolor{blue}{While SGHMC effectively reduces the computational cost of likelihood evaluations through minibatching, the introduction of the CES method within the FBNN framework addresses deeper challenges not fully mitigated by SGHMC alone. The primary concern revolves around the complex and often expensive likelihood evaluations inherent to sophisticated models. Minibatching, despite its computational benefits, introduces a trade-off in terms of accuracy. This compromise is particularly problematic in tasks requiring precise uncertainty quantification. \citet{betancourt2015} underscores this issue, highlighting how minibatching can lead to a significant loss of accuracy. Therefore the CES method offers a more comprehensive solution by leveraging the strengths of SGHMC for initial parameter exploration while strategically managing the limitations associated with minibatching and complex likelihood evaluations. This approach ensures a balanced consideration of computational efficiency and the accuracy of uncertainty estimates, justifying the adoption of the FBNN framework over solely relying on training BNNs with SGHMC.}

\subsection{Calibration -- Early stopping in Bayesian Neural Network}
By ``calibration'' we mean collecting an optimal sample of parameters to build an emulator with a reasonable level of accuracy. This is aligned with traditional calibration goals of balancing accuracy and reliability, but within a new context. Here, the calibration step involves an early stopping strategy, aimed at collecting a targeted set of posterior samples without fully converging to the target distribution. 
More specifically, we use the Stochastic Gradient Hamilton Monte Carlo (SGHMC) algorithm for a limited number of iterations to collect a small set of samples. These samples include both the model parameters ($\boldsymbol{\theta}^{(j)}$) and the outputs predicted by the model ($\mathcal{G}(\boldsymbol{X};\boldsymbol{\theta}^{(j)})$) for each sample $j$ out of a total of $J$ samples. The key focus of this training phase is not to obtain a precise approximation of the target posterior distribution, but rather collecting a small number of posterior samples as the ``training data'' for the subsequent emulation step.  
The SGHMC algorithm plays a crucial role in efficiently handling large datasets and collecting essential samples during the calibration step of the FBNN. Its ability to introduce controlled stochasticity in updates proves instrumental in preventing local minima entrapment, thereby providing a comprehensive set of posterior samples that capture the variability in the parameter space.

%During the calibration phase, we save samples of model parameters and their corresponding predictions at each iteration. The diverse set of samples obtained through SGHMC establishes a robust foundation for subsequent steps in the FBNN methodology.This strategic choice of SGHMC in the calibration step lays the groundwork for the emulation phase by contributing to the construction of a more adaptable emulator for the true neural network mapping. The broad coverage of the parameter space in the calibration step facilitates the generation of representative and diverse samples, further enhancing the overall efficiency and reliability of the FBNN methodology. In essence, the efficacy of SGHMC in exploring parameter spaces ensures a seamless transition from accurate parameter estimation to the construction of an adaptable emulator, making it a key component in the FBNN workflow.

\subsection{Emulation -- Deep Neural Network (DNN)}

The original forward mapping in BNN involves mapping input dataset \(X\) to the response variable \(Y\). For the likelihood evaluation using original forward mapping, it is necessary to calculate the likelihood $L(\boldsymbol{\theta} ; X, Y)$ for each sample of model parameters. This means that with each iteration, when a new set of model parameters is introduced, the original forward mapping needs to be applied to generate output predictions, followed by the calculation of the likelihood. In general, this process can be very time-consuming. If, however, we have a small set of estimated model parameters along with their corresponding predicted outputs collected during the calibration step, an emulator can be trained to eliminate the intermediary step (passing through each data point), allowing us to map the parameters directly to the likelihood function. This leads to a computationally efficient likelihood evaluation. Therefore, to address the computational challenges of evaluating the likelihood with large datasets, we build an emulator $\mathcal{G}^e$ using the recorded pairs $\{\boldsymbol{\theta}^{(j)},\boldsymbol{\hat{y}}^{(j)}=\mathcal{G}(\boldsymbol{X};\boldsymbol{\theta}^{(j)})\}_{j=1}^J$ obtained during the calibration step. More specifically, these input-output pairs are used to train a DNN model as an emulator $\mathcal{G}^e$ of the forward mapping $\mathcal{G}$:
\begin{eqnarray}
\mathcal{G}^{e}(\boldsymbol{X};\boldsymbol{\theta}) &=& \text{DNN}(\boldsymbol{\theta}, \mathcal{G}(\boldsymbol{X};\boldsymbol{\theta}))  = F_{K-1} \circ \cdots \circ F_0(\boldsymbol{\theta}), \\
F_k(\cdot) &=& g_k(W_k \cdot + b_k) \in C\left(\mathbb{R}^{d_k}, \mathbb{R}^{d_{k+1}}\right)
\end{eqnarray}
Given a DNN model where $\boldsymbol{\theta}$ represents the input and $\mathcal{G}(\boldsymbol{X};\boldsymbol{\theta})$ denotes the output, we set the dimensions as $d_0=d$ and $d_K=D$, where \( d \) represents the dimension of the input parameter vector \( \boldsymbol{\theta} \), and \( D \) represents the dimension of the output \( \mathcal{G}(\boldsymbol{X}; \boldsymbol{\theta}) \). %s$. 
Here, the matrices $W_k$ are defined in the space $\mathbb{R}^{d_{k+1} \times d_k}$ and the vectors $b_k$ in $\mathbb{R}^{d_{k+1}}$. The functions $g_k$ act as (continuous) activation mechanisms. In the context of our numerical examples, the activation functions for the DNN emulator are selected to ensure that both the function approximations and their derived gradients have minimized errors. This involves a grid search over a
predefined set of activation functions to ensure that the network efficiently approximates the target functions and their gradients.

After the emulator is trained, the log-likelihood can be approximated as follows:
\begin{eqnarray}
L(\boldsymbol{\theta} ; X, Y) \approx L^e(\boldsymbol{\theta} ; X, Y) = L(Y - \mathcal{G}^{e}(X;\boldsymbol{\theta}))
\end{eqnarray}
By combining the approximate likelihood $L^e(\boldsymbol{\theta} ; X, Y)$ with the prior probability $p(\boldsymbol{\theta})$, an approximate posterior distribution can be obtained. 
Similarly, we could approximate the potential function using the predictions from DNN:
\begin{eqnarray}
\Phi(\boldsymbol{\theta}; \boldsymbol{X}, \boldsymbol{Y}) \approx \Phi^e(\boldsymbol{\theta}; \boldsymbol{X}, \boldsymbol{Y}) = \frac{1}{2}\|\boldsymbol{Y} - \mathcal{G}^{e}(X;\boldsymbol{\theta})\|_{\Gamma}^2
\end{eqnarray}

Building upon the foundational concepts of using a DNN emulator \(\mathcal{G}^e\)  for approximating the forward mapping function \(\mathcal{G}\), we further elaborate on the implications and advantages of this approach for Bayesian inference, particularly in the context of handling large datasets and/or complex likelihood functions. The emulation step, which involves training the DNN emulator with input-output pairs \(\{\boldsymbol{\theta}^{(j)},\mathcal{G}(\boldsymbol{X};\boldsymbol{\theta}^{(j)})\}\), serves as a critical phase where the emulator learns to mimic the behavior of the original model with high accuracy. The utilization of DNN emulator to approximate the likelihood function in Bayesian inference presents a significant computational advantage over the direct use of the original BNN likelihood. This advantage stems primarily from the inherent differences in computational complexity between evaluating the  the likelihood with a DNN emulator -- which takes a set of model parameters as input and yields predicted responses—and the original BNN model -- which processes $\boldsymbol{X}$ as input to produce the response variable.

In the sampling stage, the computational complexity could be significantly reduced if we use $\Phi^e$ instead of $\Phi$ in the accept/reject step of MCMC. If the emulator is a good representation of the forward mapping, the difference between $\Phi^e$ and $\Phi$ would be small and negligible. Then, the samples by such emulative MCMC have the stationary distribution that closely follows the true posterior distribution. This approach not only ensures that the sampling process is computationally feasible, but also maintains the integrity of the stationary distribution, closely approximating the true posterior distribution with minimal discrepancy. The integration of DNN emulators into the Bayesian inference workflow thus presents a compelling solution to the computational challenges associated with evaluating likelihood functions in complex models.

\subsection{Sampling -- Preconditioned Crank-Nicolson (pCN)}
In the context of the FBNN method, the sampling step is crucial for approximating the posterior distribution efficiently. The method employs MCMC algorithms based on a trained emulator to achieve full exploration and exploitation. However, challenges arise, especially in high-dimensional parameter spaces, where classical MCMC algorithms often exhibit increasing correlations among samples.
To address this issue, the pCN method presented in Algorithm \ref{alg:pcn} has been used in our proposed framework as a potential solution. Unlike classical methods, pCN avoids dimensional dependence challenges, making it particularly suitable for scenarios like BNN models with a high number of weights to be inferred \citep{hairer2009}.

As explained in section \ref{subsec:pcn}, the pCN approach minimizes correlations between successive samples, a critical feature for ensuring the representativeness of the samples collected. This characteristic is vital for FBNNs, as it directly impacts the network's ability to learn from data and make robust predictions. The pCN method excels in traversing the parameter space with controlled perturbations, enhancing the algorithm's ability to capture the most probable configurations of model parameters. This focus on effective exploration around the mode contributes to a more accurate representation of the underlying neural network, ultimately improving model performance. In other words, the choice of pCN as the sampling method in FBNN is motivated by its tailored capacity to navigate and characterize the most probable regions of the parameter space. This choice reinforces the methodology's robustness and reliability, as pCN facilitates efficient sampling, leading to a more accurate and representative approximation of the posterior distribution.

To illustrate this, Figure \ref{fig:gp_sim} displays a simulation that contrasts the sampling mechanisms of SGHMC and pCN within a multimodal probability distribution. The task is to sample from a mixture of 25 Gaussian distributions, represented in panel (a), using a total of 200,000 samples. Here, the target distribution is multimodal with several distinct peaks (modes). Middle figure shows that SGHMC has explored the parameter space, although with a less concentrated sampling around the modes compared to the target distribution. This indicates that while SGHMC is effective at exploring the space, it may not capture the modes as tightly as the target distribution. In the right panel related to pCN sampler, the concentration of samples around the modes is much higher compared to SGHMC, which indicates that pCN is more effective at exploring around the modes of the distribution. Thus, we believe the combination of SGHMC and pCN in our proposed framework can complement each other for a more effective exploration of the parameter space. 

\begin{figure}[t!]
    \centering
    \includegraphics[width=0.95\textwidth]{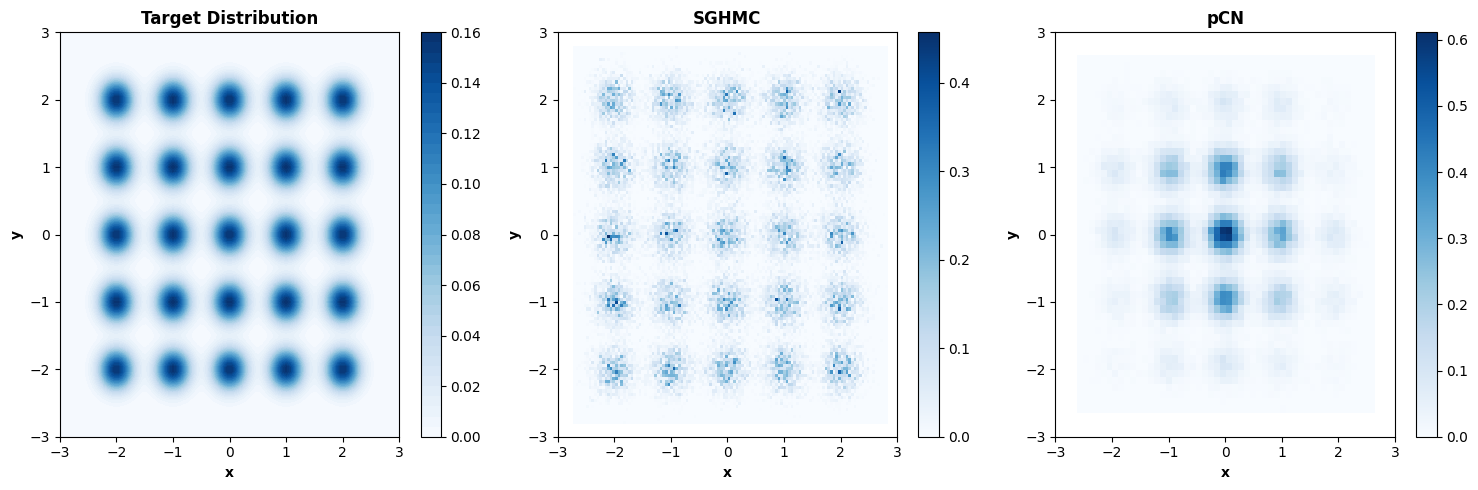}
    \caption{Sampling from a mixture of 25 Gaussians shown in (a) with 200k samples. SGHMC in (b) broadly explores the space, while pCN in (c) hones in on the high-density regions for precise mode capture.}
    \label{fig:gp_sim}
\end{figure}

\subsection{Theoretical Foundations}

In this section, we aim to quantify the error between the true potential function and its emulation in the context of the FBNN method.
Let $\Omega=(0,1)^d$ and consider forward mappings in the Sobolev space $W^{n,p}(\Omega):=\{f
\in L^p(\Omega) : D^\alpha f\in L^p(\Omega) \, \text{ for all }\, \alpha \in (\mathbb{N}\cup\{0\})^d \, \text{with} \, |\alpha|\leq n\}$ with $\Vert f\Vert_{n,p}:= \left(\sum_{0\leq |\alpha|\leq n}\Vert D^\alpha f\Vert_p^p\right)^{\frac{1}{p}}$, and $\Vert f\Vert_{n,\infty}:=\max_{0\leq |\alpha|\leq n} \Vert D^\alpha f\Vert_\infty$.
Define the Sobolev-Slobodeckij norm for $0<s<1$: $\Vert f\Vert_{s,p}:=\left(\Vert f\Vert_p^p + \int_\Omega\int_\Omega \frac{|f(x)-f(y)|^p}{|x-y|^{sp+d}}dxdy\right)^{\frac{1}{p}}$ and $\Vert f\Vert_{s,\infty}:=\max\left\{\Vert f\Vert_\infty, \mathrm{ess}\sup_{x,y\in\Omega} \frac{f(x)-f(y)}{|x-y|^s}\right\}$.

\begin{thm}\label{thm:err_bound}
Let $1\leq p\leq \infty$ and $0\leq s< 1$. Assume $\mathcal{G}_j(X;\,\cdot)\in W^{n,p}(\Omega)\cap L^{\infty}(\Omega)$ for $j=1,\cdots, D$. For any $\epsilon\in(0,1/2)$, there is a standard NN, $\mathcal{G}^e$, with ReLU activation functions such that
\begin{equation}
    \Vert \Phi-\Phi^e\Vert_{s,p} \leq \epsilon .
\end{equation}
and the depth $K\leq c\log(\epsilon^{-n/(n-s)})$, the number of weights and units $N\leq c\epsilon^{-d/(n-s)}\log^2(\epsilon^{-n/(n-s)})$ with constant $c=c(d,n,p,s)>0$.
\end{thm}
\begin{proof}
Note that we have
\begin{equation*}
\begin{aligned}
    \Phi(\boldsymbol{\theta})-\Phi^e(\boldsymbol{\theta}) = \frac{1}{2} [\langle\mathcal{G}(X;\boldsymbol{\theta})-\mathcal{G}^e(X;\boldsymbol{\theta}), y-\mathcal{G}(X;\boldsymbol{\theta})\rangle_{\Gamma} +  \langle y-\mathcal{G}^e(X;\boldsymbol{\theta}),\mathcal{G}(X;\boldsymbol{\theta})-\mathcal{G}^e(X;\boldsymbol{\theta})\rangle_{\Gamma}]
\end{aligned}
\end{equation*}

Because $\mathcal{G}_j(X;\,\cdot)\in L^{\infty}((0,1)^d)$, there exists a constant $M>0$ such that $\max_{1\leq j\leq D}\Vert \mathcal{G}_j(X;\,\cdot)\Vert_\infty, \Vert y\Vert \leq M$.
For $\epsilon/(MD)>0$, by Theorem 4.1 of \citet{guhring2020error}, there exists a standard NN with ReLU activation functions and the depth $K$ and the number of weights and units as in the condition such that
\begin{equation*}
    \Vert \mathcal{G}_j(X;\,\cdot)-\mathcal{G}_j^e(X;\,\cdot)\Vert_{s,p} \leq \epsilon/(MD), \quad j=1,\cdots, D.
\end{equation*}
Therefore we have
\begin{equation*}
\begin{aligned}
    \Vert \Phi-\Phi^e\Vert_{s,p} &\leq M \sum_{j=1}^D \Vert \mathcal{G}_j(X;\,\cdot)-\mathcal{G}_j^e(X;\,\cdot)\Vert_{s,p} \leq \epsilon.
\end{aligned}
\end{equation*}
\end{proof}

Denote the Hellinger distance between densities as $d_H(\pi, \pi^e) = \int (\sqrt{\pi}-\sqrt{\pi^e})^e d\mu$. Then we describe how the emulation error propogates into the Hellinger error in the likelihood.
\begin{thm}
Let $\pi(\cdot\,;\boldsymbol{\theta})\propto \exp(-\Phi(\cdot\,;\boldsymbol{\theta}))$ and $\pi^e(\cdot;\boldsymbol{\theta})\propto \exp(-\Phi^e(\cdot\,;\boldsymbol{\theta}))$ denote the likelihood and its emulation, respectively. Suppose the conditions of Theorem \ref{thm:err_bound} holds for $p=\infty$. Then we have
\begin{equation}
    d_H(\pi, \pi^e) \lesssim \epsilon .
\end{equation}
where $\epsilon$ satisfies the constraints in Theorem \ref{thm:err_bound}.
\end{thm}
\begin{proof}
Compute the Hellinger distance
\begin{equation*}
\begin{aligned}
2 d_H^2(\pi, \pi^e) &= \int (\sqrt \pi-\sqrt{\pi^e})^2 d\mu = \int \left[1-\exp\left(\frac{1}{2}\Phi(y;\boldsymbol{\theta})-\frac{1}{2}\Phi^e(y;\boldsymbol{\theta})\right)\right]^2 \pi(y;\boldsymbol{\theta}) dy \\
&\leq \int \frac{C}{4} |\Phi(y;\boldsymbol{\theta})-\Phi^e(y;\boldsymbol{\theta})|^2 \pi(y;\boldsymbol{\theta}) dy \leq C' \int \Vert \Phi-\Phi^e\Vert_{\infty}^2 \pi(y;\boldsymbol{\theta}) dy \lesssim \epsilon^2 .
\end{aligned}
\end{equation*}
where the last inequality is by Theorem \ref{thm:err_bound}. Taking square-root on both sides yields the conclusion.
\end{proof}

Furthermore, given Gaussian prior for $\boldsymbol{\theta}$, $\pi_0$, we can characterize the discrepancy of posteriors with the original and emulated likelihoods in the following theorem.
\begin{thm}
Let $\hat\pi(\boldsymbol{\theta})\propto \exp(-\Phi(y;\boldsymbol{\theta}))\pi_0(\boldsymbol{\theta})$ and $\hat\pi^e(\boldsymbol{\theta})\propto \exp(-\Phi^e(y;\boldsymbol{\theta}))\pi_0(\boldsymbol{\theta})$ denote the posterior and the one with emulated likelihood, respectively. Suppose the conditions of Theorem \ref{thm:err_bound} holds for $p=\infty$. Then we have
\begin{equation}
    d_H(\hat\pi, \hat\pi^e) \lesssim \epsilon .
\end{equation}
where $\epsilon$ satisfies the constraints in Theorem \ref{thm:err_bound}.
\end{thm}
\begin{proof}
By Bayes' theorem, we have
\begin{equation*}
\begin{aligned}
    \hat\pi(\boldsymbol{\theta}) &= \frac{1}{Z(y)}\exp(-\Phi(y;\boldsymbol{\theta}))\pi_0(\boldsymbol{\theta}), \quad 0< Z(y)= \int_{\Omega} \exp(-\Phi(y;\boldsymbol{\theta}))\pi_0(\boldsymbol{\theta}) d\boldsymbol{\theta}) <+\infty \\
    \hat\pi^e(\boldsymbol{\theta}) &= \frac{1}{Z(y)}\exp(-\Phi^e(y;\boldsymbol{\theta}))\pi_0(\boldsymbol{\theta}), \quad 0< Z^e(y)= \int_{\Omega} \exp(-\Phi^e(y;\boldsymbol{\theta}))\pi_0(\boldsymbol{\theta}) d\boldsymbol{\theta}) <+\infty .
\end{aligned}
\end{equation*}
Therefore we compute the Hellinger distance
\begin{equation*}
\begin{aligned}
d_H(\hat\pi, \hat\pi^e) &= \int_\Omega \left[Z(y)^{-\frac{1}{2}}\exp\left(-\frac{1}{2}\Phi(y;\boldsymbol{\theta})\right) - Z^e(y)^{-\frac{1}{2}}\exp\left(-\frac{1}{2}\Phi^e(y;\boldsymbol{\theta})\right)\right]^2 \pi_0(d\boldsymbol{\theta}) \\
&\leq \frac{2}{Z(y)} \int_\Omega \left[\exp\left(-\frac{1}{2}\Phi(y;\boldsymbol{\theta})\right) -\exp\left(-\frac{1}{2}\Phi^e(y;\boldsymbol{\theta})\right)\right]^2 \pi_0(d\boldsymbol{\theta}) + 2|Z(y)^{-\frac{1}{2}}-Z^e(y)^{-\frac{1}{2}}|^2 Z^e(y) \\
&\leq 2\int_\Omega \left[1-\exp\left(\frac{1}{2}\Phi(y;\boldsymbol{\theta})-\frac{1}{2}\Phi^e(y;\boldsymbol{\theta})\right)\right]^2 \hat\pi(d\boldsymbol{\theta}) + C |Z(y)-Z^e(y)|^2 \\
&\leq \int \frac{C'}{2} |\Phi(y;\boldsymbol{\theta})-\Phi^e(y;\boldsymbol{\theta})|^2 \hat\pi(d\boldsymbol{\theta}) + C \int_\Omega |\exp(-\Phi(y;\boldsymbol{\theta})) -\exp(-\Phi^e(y;\boldsymbol{\theta}))|^2 \pi_0(d\boldsymbol{\theta}) \\
&\leq \frac{C'}{2} \Vert \Phi-\Phi^e\Vert_\infty^2 + C'' \Vert \Phi-\Phi^e\Vert_\infty^2 \lesssim \epsilon^2
\end{aligned}
\end{equation*}
where the last inequality is by Theorem \ref{thm:err_bound}. Taking square-root on both sides yields the conclusion.
\end{proof}

\section{Numerical Experiments}

\subsection{Setup}
We demonstrate the effectiveness of our method on eleven synthetic and real-world datasets, comparing it against a comprehensive selection of baseline approaches.

\paragraph{Datasets.} We experiment on a series of regression and classification problems. Detailed information regarding these datasets, including the number of features and datapoints in each, and the number of parameters used in the main FBNN model for each dataset, is outlined in Table \ref{tab:datasets}. We have also included the details of the DNN Emulator architecture for each dataset in Table \ref{tab:dnn_emulator}.

\begin{table}[t!]

  \centering
  \caption{Description of various datasets used to evaluate the overall performance of our proposed approach against state-of-the-art baseline methods.}
  \label{tab:datasets}
\begin{tblr}{
    colspec={X[0.8,c] X[1.5,c] X[c] X[c] X[1.6,c]},  % Adjust the numbers in X[2] and X[3] as needed
    row{odd}={bg=lightgray},
    cell{odd}{1}={bg=white},
    row{1}={bg=RoyalBlue, fg=white},
    row{1}={font=\bfseries},
    hline{7,13} = {},
    vlines,
}
    Task & Dataset & {\# Datapoints} & {\# Features} & {\# FBNN parameters} \\
    Regression & Boston Housing & 506 & 13 & 3,009 \\
    & Wine Quality & 1,599 & 11 & 241 \\
    & Alzheimer & 185,831 & 56 & 33,345 \\
    & Year Prediction MDS & 515,345 & 90 & 81,901 \\
    & Simulation & 5,000,000 & 1,000 & 52,832 \\
    Classification & Adult & 40,434 & 14 & 2,761 \\
    & Mnist & 70,000 & 784 & 3,961 \\
    & Alzheimer & 185,831 & 56 & 33,345 \\
    & celebA & 202,599 & 39 & 1,521 \\
    & SVHN & 630,420 & 3072 & 171,313 \\
    & Simulation & 5,000,000 & 1,000 & 52,832 \\
  \end{tblr}

\vspace{12pt}

  \centering
  \caption{Description of various datasets and their corresponding DNN Emulator architectures (Droupout layers have been used on input layer and first hidden layer)}
  \label{tab:dnn_emulator}
\begin{tblr}{
    colspec={X[1.2,c] X[1.6,c] X[c] X[1.4,c] X[1.4,c] X[c] X[c] X[2,c]},
    row{odd}={bg=lightgray},
    cell{odd}{1}={bg=white},
    row{1}={bg=RoyalBlue, fg=white},
    row{1}={font=\bfseries},
    hline{7,13} = {},
    vlines
}
    Task & Dataset & {\#Hidden Layers} & {\# Neurons per Layer} & Activation Functions & \#Epochs & Dropout Rate\\
    Regression & Boston Housing & 2 & 3,32 & ReLU & 1000 & 0.7 \\
    & Wine Quality & 2 & 3,32 & ReLU & 1000 & 0.5 \\
    & Alzheimer & 2 & 4,64 & ReLU & 1000 & 0.5 \\
    & Year Prediction & 2 & 4,64 & ReLU & 1000 & 0.5 \\
    & Simulation & 3 & 8,64,32 & ReLU & 1000 & 0.5 \\
    Classification & Adult & 2 & 4,32 & ReLU & 1000 & 0.5 \\
    & Mnist & 2 & 3,64 & ReLU & 1000 & 0.5 \\
    & Alzheimer & 2 & 4,64 & ReLU & 1000 & 0.5 \\
    & celebA & 2 & 3,32 & ReLU & 1000 & 0.5 \\
    & SVHN & 3 & 4,64,32 & ReLU & 1000 & 0.7 \\
    & Simulation & 3 & 8,64,32 & ReLU & 1000 & 0.5 \\
  \end{tblr}
\end{table}

\paragraph{Baseline Methods.} We present empirical evidence comparing our CES method against a broad array of baseline approaches including two baseline BNN methods equipped with the SGHMC and pCN samplers (shown as BNN-SGHMC and BNN-pCN), and BNN architectures incorporating Variational Inference, Lasso Approximation, MC-Dropout, SWAG, and RNS-HMC. Detailed information about these methods was provided in Section \ref{sec:methods}. We also include the results from DNN, which does not provide uncertainty quantification, but serves as a reference point. Moreover, we provide the results of Deep Ensembles, which consist of multiple DNNs, each initialized with different random seeds. We refer to this method as Ensemble-DNN. Although the Ensemble-DNN approach allows for parallelization, it falls short in providing a probabilistic framework for analysis, a significant advantage offered by our CES method.

As discussed earlier, one of the distinctive features of our main FBNN model, more specifically shown as FBNN (SGHMC-pCN), lies in its strategic integration of the SGHMC sampler during the calibration step and the pCN algorithm during the sampling step. This combination is carefully chosen to harness the complementary strengths of these two sampling methods. Further expanding our exploration, we introduce three additional FBNN models: FBNN (pCN-SGHMC), where pCN is employed in the calibration step and SGHMC in the sampling step; FBNN (pCN-pCN), where pCN is used in both steps; and FBNN (SGHMC-SGHMC), where SGHMC is used in both calibration and sampling steps. 

Throughout these experiments, we collect 2000 posterior samples for the BNN-SGHMC and BNN-pCN, with samples being collected at each iteration. In contrast, for the FBNN methods, we use a small number (200) of samples from either BNN-SGHMC or BNN-pCN (depending on the specific FBNN model) along with the corresponding predicted outputs during the calibration step. These 200 samples serve as the ``training data'' for the emulator. 
Moreover, we evaluate the efficacy of utilizing only the initial 200 samples from the BNN-SGHMC model across all the datasets. This was done to demonstrate the necessity of collecting more samples, either using the original BNN or employing the FBNN method, rather than relying our inference on a limited number of initial samples.
 
It is also crucial to highlight that the BNN-SGHMC and BNN-pCN models are trained from a randomly chosen initial point for the MCMC sampling process. On the other hand, in the FBNN methods, we employ the set of posterior samples collected during the last iteration of the calibration step as the starting point for the subsequent MCMC sampling.

\paragraph{Metrics.} To thoroughly assess the performance and effectiveness of each method, we use a range of key metrics. These include Mean Squared Error (MSE) for regression tasks (Figure \ref{fig:mse}) and Accuracy for classification tasks (Figure \ref{fig:accuracy}). We also evaluate the models based on their computational cost, and various statistics related to the Effective Sample Size (ESS) of model parameters. These statistics include the minimum, maximum, and median ESS, as well as the minimum ESS per second. We also quantify the amount of speedup, denoted as ``spdup'', a metric that compares the minimum ESS per second of each model with that of BNN-SGHMC as the benchmark (Figure \ref{fig:spdup}). Analysing spdup is crucial as it provides a comparative measure of efficiency, highlighting the model's capability to achieve high-quality parameter sampling with lower computational resource utilization relative to the benchmark BNN-SGHMC.

The effective sample size takes the autocorrelation among the consecutive samples into account. While we can reduce autocorrelation using the thinning strategy, this leads to a higher computational time for the same number of samples. Our spdup metric allows for a fair comparison of sampling algorithms (regardless of what thinning strategy used) by taking both autocorrelation and computational cost into account.

For UQ in regression cases, we evaluate the Coverage Probability (CP) set at 95\%. In addition, we construct 95\% Credible Intervals (CI) by the prediction results of Bayesian models, along with the average true output, to illustrate UQ in regression problems. For classification problems, we use Expected Calibration Error (ECE) and Reliability Diagrams to evaluate UQ. ECE addresses model calibration, aiming for accurate uncertainty estimates, while reliability diagrams offer a visual summary of probabilistic forecasts.

Figures \ref{fig:mse}, \ref{fig:accuracy}, and \ref{fig:spdup} summarize our results. More detailed results are provided in Tables \ref{tab:performance-regression} and \ref{tab:performance-classification}.

\begin{figure}[t!]
    \centering
    \includegraphics[width=1\textwidth]{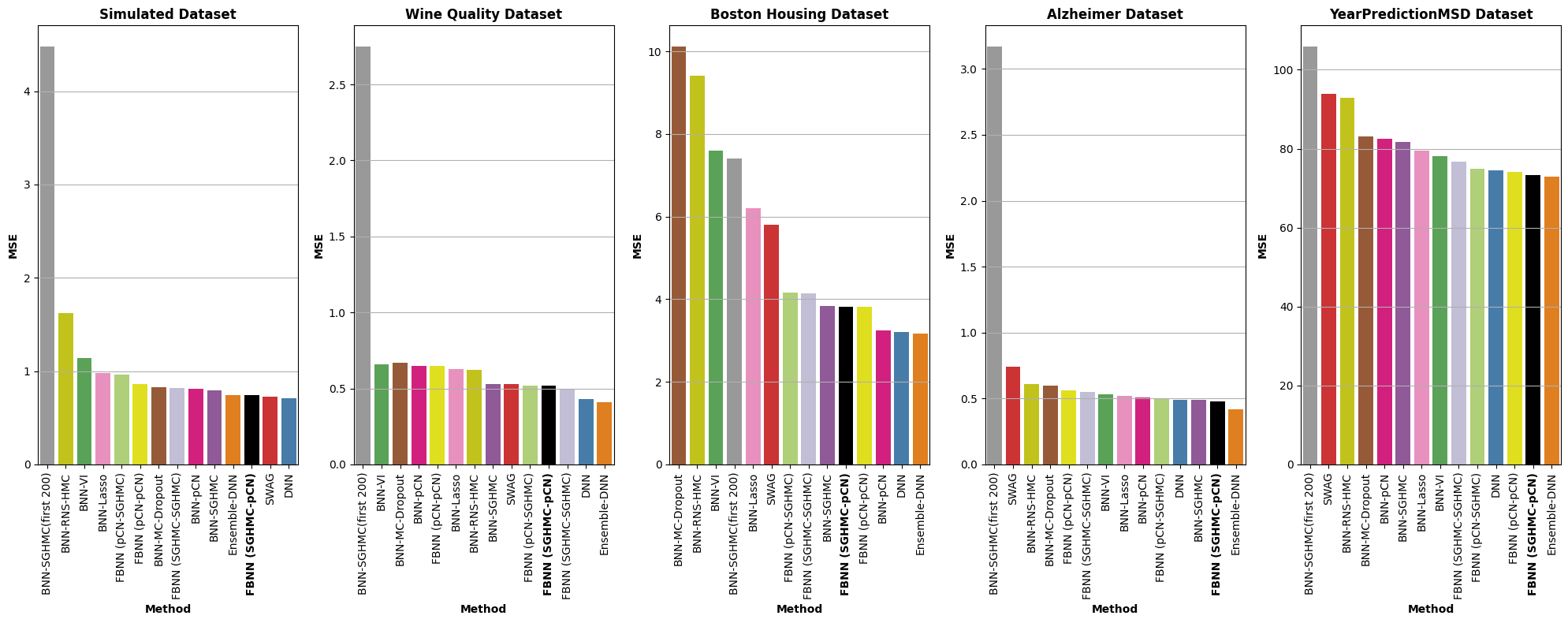}
    \caption{A comprehensive comparison of the MSE for various BNN methods across five regression datasets. Each subplot corresponds to a dataset and the color-coded bars represent the distinct methodologies evaluated. The main FBNN method, called FBNN (SGHMC-pCN), highlighted in bold, shows among the lowest values for MSE among all datasets.}
    \label{fig:mse}
    \vspace{6pt}

    \centering
    \includegraphics[width=1\textwidth]{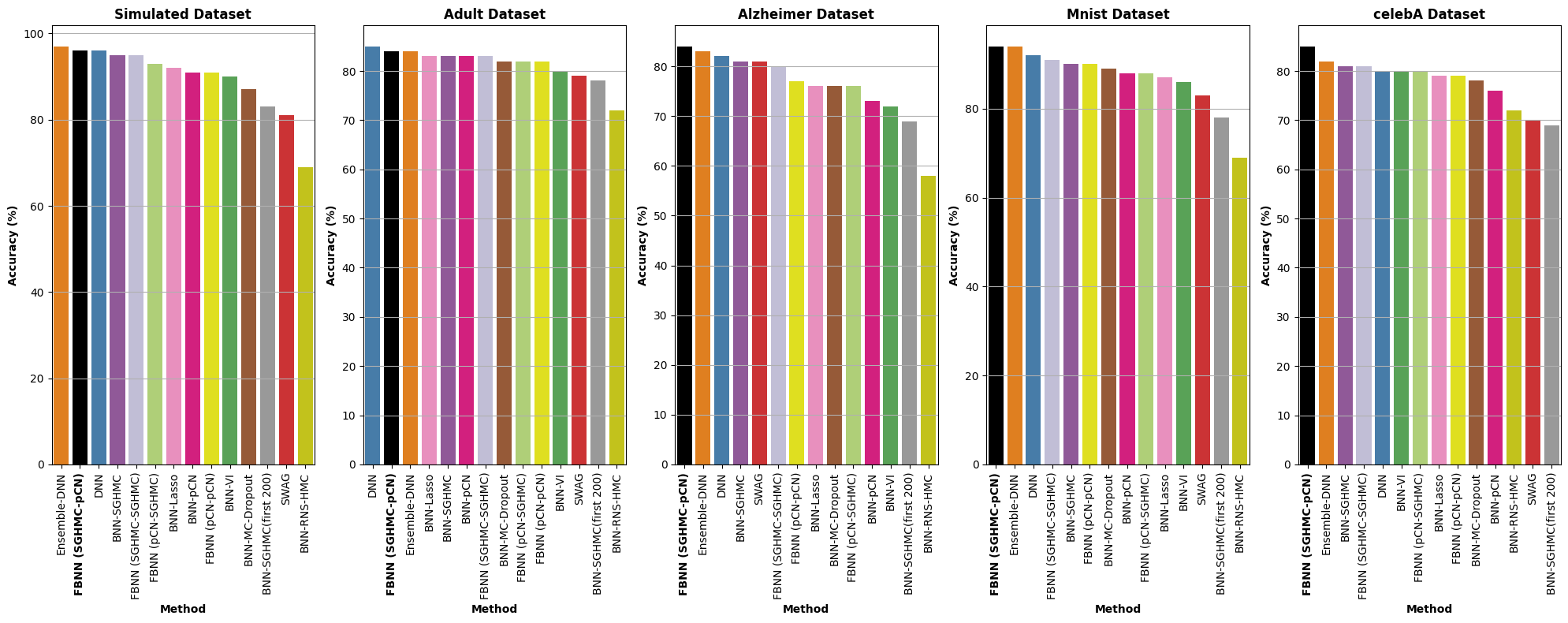}
    \caption{A comprehensive comparison of the Accuracy for various BNN methods across five classification datasets. Each subplot corresponds to a dataset and the color-coded bars represent the distinct methodologies evaluated. The main FBNN method, FBNN (SGHMC-pCN), highlighted in bold, demonstrates superior performance by achieving the highest accuracy in four out of the five datasets examined and the second highest in one dataset. }
    \label{fig:accuracy}
    \vspace{6pt}

\end{figure}

\subsection{Regression Tasks} 

We first evaluate our proposed method using a set of simulated and real regression problems. %More detailed results are provided in Table \ref{tab:performance-regression}.

\begin{figure}[t!]
    \centering
    \includegraphics[width=1\textwidth]{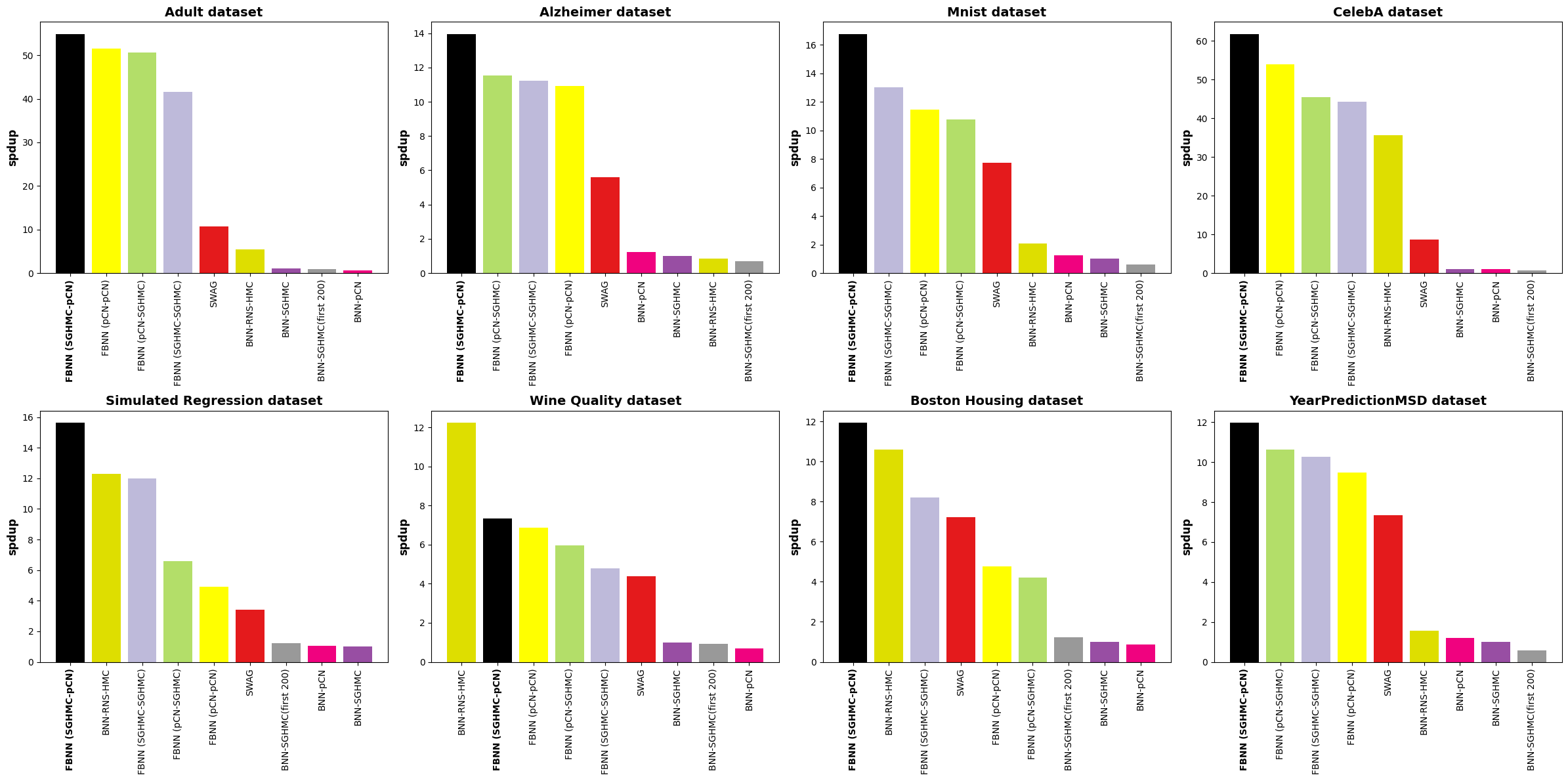}
    \caption{Comparative analysis of speedup (spdup) for various Bayesian Neural Network (BNN) methods across tested datasets. Methods are ordered by efficiency within each dataset, highlighting the impact of model characteristics on sampling performance. Notably, "FBNN (SGHMC-pCN)" achieves the highest speedup among all datasets except for the Wine Quality dataset, where it ranks as the second highest, underscoring its exceptional efficiency in diverse analytical contexts.}
    \label{fig:spdup}
    \vspace{-8pt}
\end{figure}

\paragraph{Simulated Data.}

We begin our empirical evaluation by using simulated data. 
To this end, we utilize the \verb|make_regression| function from the \verb|sklearn.datasets| package to generate a dataset consisting of 5,000,000 observations and 1,000 predictors.

Figure \ref{fig:loglikelihood} compares the true and emulated log likelihood functions associated with posterior samples collected using BNN-SGHMC. The emulated values are based on the FBNN (SGHMC-pCN) model. As we can see, the two functions are similar, indicating that the emulator provides a reasonable approximation of the true target distribution.

Figure \ref{fig:mse} compares the MSE among all models, showing that while the DNN method achieves the lowest MSE at 0.71, the FBNN (SGHMC-pCN) model provides a similar performance. Notably, among all the FBNN variants, FBNN (SGHMC-pCN) provides the highest CP at 92.2\%, demonstrating a level of calibration comparable to that of the BNN model. The Ensemble-DNN demonstrates comparable performance to FBNN (SGHMC-pCN) in terms of CP, yet it operates at a pace three times slower. %The complete results are provided in Table \ref{tab:performance-regression}. 

Examining the efficiency of sample generation, all FBNN variants have relatively higher ESS per second compared to all the other BNN models, except for BNN-RNS-HMC. Among all the models, FBNN (SGHMC-pCN) has the highest min ESS per second at 0.043. Figure \ref{fig:spdup} indicates our model provides the highest speedup (16.33) compared to BNN-SGHMC as the baseline model, highlighting our method's computational efficiency.
%The observed longer computation times for FBNN (pCN-pCN) and FBNN (pCN-SGHMC) in comparison to FBNN (SGHMC-pCN) and FBNN (SGHMC-SGHMC) for the simulated dataset is attributed to the pCN algorithm's slower convergence to a stable state during the calibration phase. Due to this slower convergence, the first 200 samples collected during the calibration process did not meet the desired quality standards and, as a result, were not considered qualified enough to be used for training the emulator. Recognizing the need for a refined calibration approach, a strategic decision was made to prioritize the subsequent 200 samples where the pCN algorithm demonstrated improved convergence and produced higher-quality calibration data. The requirement to collect 400 samples during the calibration step for FBNN (pCN-pCN) and FBNN (pCN-SGHMC) models, as opposed to the 200 samples needed for FBNN (SGHMC-pCN) and FBNN (SGHMC-SGHMC) models, resulted in decreased computational efficiency for the former.
Considering these results, FBNN (SGHMC-pCN) emerges as a strong approach with a good balance between predictive accuracy, computational efficiency, and uncertainty quantification, making it the overall best option for Bayesian deep learning.

Figure \ref{fig:ciwine_all}a shows the estimated mean and prediction uncertainty for both BNN-SGHMC and FBNN (SGHMC-pCN) models, alongside the smoothed average and 95\% interval for the true output. For clarity and conciseness within our figures, we have employed Principal
Component Analysis (PCA) and used the first principal component to transform the original data into a one-dimensional representative feature (x-axis in Figure \ref{fig:ciwine_all}). As we can see, BNN and FBNN have very similar credible intervals. This consistency in credible interval bounds is significant for UQ, indicating that both models effectively and almost equally quantify uncertainty in their predictions.

\begin{figure}[t!]
    \centering
    \includegraphics[width=0.95\textwidth]{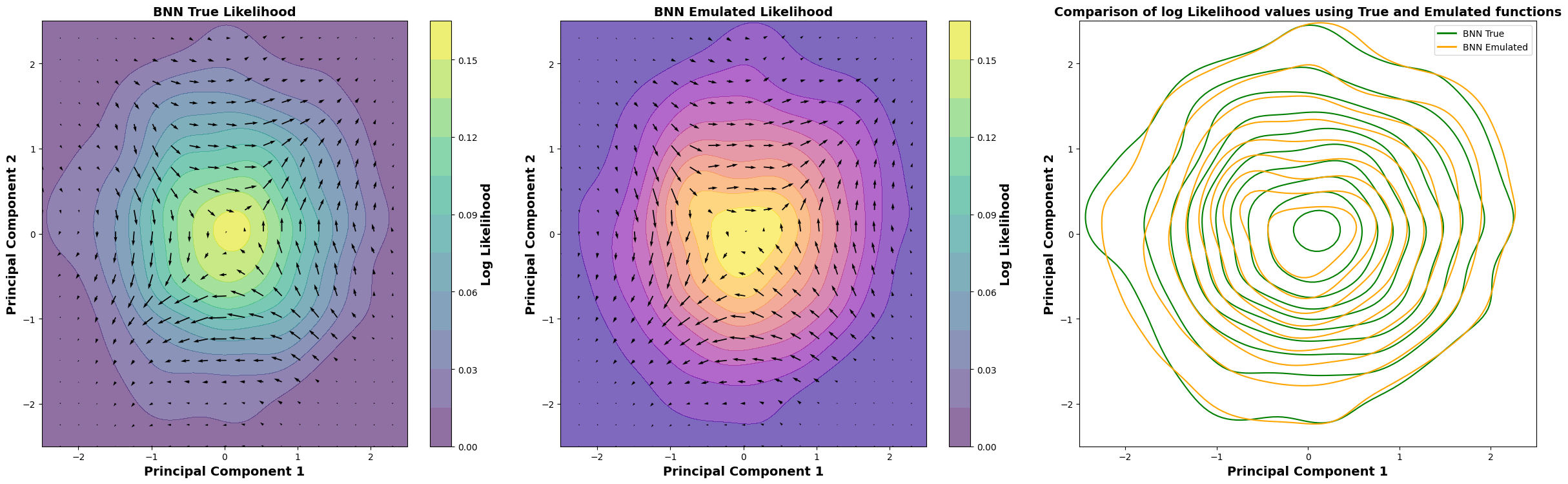}
    \caption{Evaluating the performance of the emulator. The plot contrasts log likelihood values obtained using the true likelihood function against those derived using the emulated likelihood function. The x-axis and y-axis represent the first and second principal components of the model parameters based on the MCMC samples obtained in our simulation study.}
    \label{fig:loglikelihood}
    \vspace{12pt}

    \centering
    \begin{subfigure}{.32\textwidth}
        \centering
        \includegraphics[width=\linewidth]{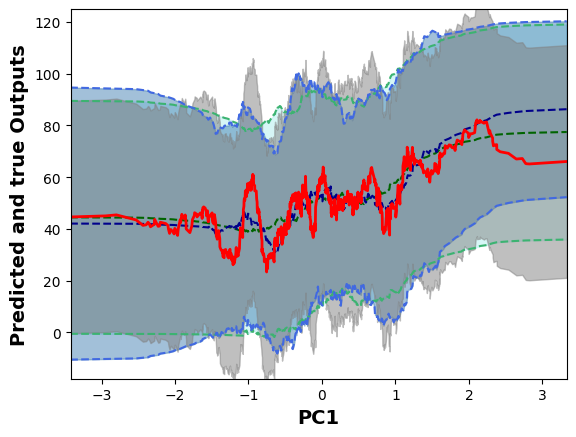}
        \caption{Simulated data}
        \label{fig:cisimulate}        
    \end{subfigure}%
    \hspace{3pt}
    \begin{subfigure}{.32\textwidth}
        \centering
        \includegraphics[width=\linewidth]{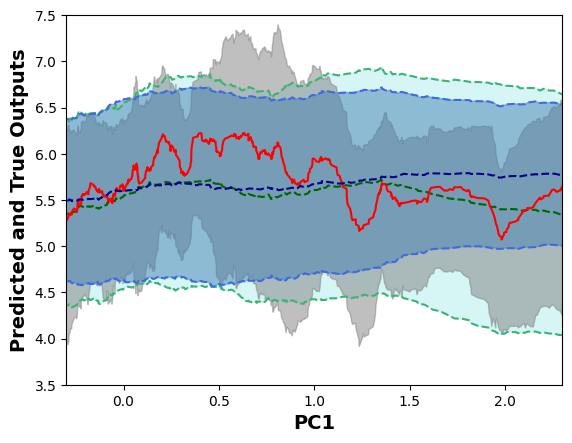}
        \caption{Wine Quality data}
        \label{fig:ciwine}
    \end{subfigure}%
    \hspace{3pt}
    \begin{subfigure}{.32\textwidth}
        \centering
        \includegraphics[width=\linewidth]{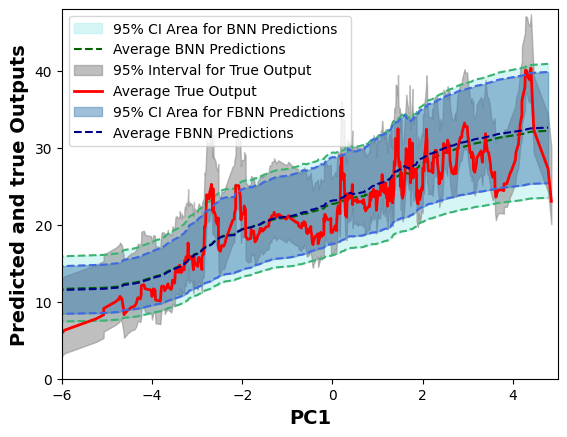}
        \caption{Boston Housing data}
        \label{fig:ciboston}
    \end{subfigure}
    \caption{Comparative analysis of predictive credible intervals and mean predictions for regression tasks.  For each dataset, the 95\% CI for BNN predictions and FBNN predictions are shown as shaded areas. The average predictions from BNN and FBNN are represented with dashed lines. Additionally, the 95\% CI for the true output as ground truth and the smoothed average true output are plotted as solid lines. The  x-axis shows the first principal component of the predictors. }
    \label{fig:ciwine_all}
\end{figure}

\paragraph{Wine Quality Data.}
As the first real dataset for the regression task, we use the Wine Quality data \citep{cortez2009}. This dataset contains various physicochemical properties of different wines, while the target variable is the quality rating. The performance of FBNN (SGHMC-pCN) indicates a well-balanced approach, making it superior to the other models for several reasons. Firstly, it achieves a competitively low MSE of 0.52, comparable to other high-performing models like BNN-SGHMC and SWAG, but it surpasses them in terms of speedup. Moreover, FBNN (SGHMC-pCN) exhibits a robust predictive performance, with a competitively high CP.
Figure \ref{fig:ciwine_all}b shows the prediction mean and 95\% CI BNN-SGHMC and FBNN (SGHMC-pCN), as well as the smoothed average and 95\% interval for the true output.

\paragraph{Boston Housing Data.}

The Boston housing dataset was collected in 1978 \citep{harrison1978}. Each of the entries present aggregated data for homes from various suburbs in Boston. For this dataset, FBNN (SGHMC-pCN) stands out with a notable balance between MSE (3.82), CP (81.1\%), and computational efficiency, completing the task in just 91 seconds. This model significantly outperforms all the other models in terms of speedup (11.94), showcasing its effectiveness in sampling.
Figure \ref{fig:ciwine_all}c shows the 95\% CIs and mean predictions of both BNN-SGHMC and FBNN (SGHMC-pCN). The FBNN (SGHMC-pCN) model, in particular, displays well-calibrated uncertainty quantification, mirroring the performance of the BNN models, implying that its probabilistic predictions capture the model uncertainty.

%In the evaluation of log probability trajectories for the Boston Housing dataset, a compelling contrast emerges between FBNN (SGHMC-pCN) and BNN-SGHMC for their respective last 1000 collected samples. Both models exhibit a decreasing and noisy pattern, indicative of the iterative nature of the sampling process. However, the log probability trajectory for BNN-SGHMC is characterized by a higher degree of noise and a less pronounced decrease compared to FBNN (SGHMC-pCN). This disparity suggests that FBNN (SGHMC-pCN) showcases a more discernible decrease during its convergence process, potentially reflecting challenges in efficiently exploring and narrowing down the parameter space.

\paragraph{Alzheimer Data.}

Next, we analyze the data from the National Alzheimer’s Coordinating Center (NACC), which is responsible for developing and maintaining a database of patient information collected from the Alzheimer disease centers (ADCs) funded by the National Institute on Aging (NIA) \citep{beekly2004national}. The NIA appointed the ADC Clinical Task Force to determine and define an expanded, standardized clinical data set, called the Uniform Data Set (UDS). The goal of the UDS is to provide ADC researchers a standard set of assessment procedures to identify Alzheimer's disease \citep{beekly2007national}. We have used 56 key features for our analysis. These features were carefully selected to represent a wide spectrum of variables relevant to Alzheimer's disease diagnosis, including functional abilities, brain morphometrics, living situations, and demographic information \citep{ren2022hierarchical}. For the regression case, the goal is to predict Left Hippocampus Volume, a critical marker in the progression of the disease \citep{vanderFlier2009}, as a function of other variables.
For this dataset, Figure \ref{fig:mse} shows that the FBNN (SGHMC-pCN) model stands out for its balanced performance, recording the second lowest MSE at 0.48 and a relatively high CP at 91.6\%. It shows a considerable improvement in computational efficiency, evidenced by a speedup factor of 22 times compared to BNN-SGHMC as the baseline BNN model.

%For the classification case, the goal is to predict Cognitive status (A diagnosis of cognitively unimpaired (healthy controls, HC) as class 0, and either mild cognitive impairment (MCI) due to AD or dementia due to AD as class 1).

\paragraph{Year Prediction MSD Data.}
For this data, the goal is to predict the release year of a song from audio features. Songs are mostly western, commercial tracks ranging from 1922 to 2011, with a peak in the year 2000s \citep{year_prediction_msd_203}.
In the context of the YearPredictionMSD dataset, FBNN (SGHMC-pCN) showcases its superiority over other models by achieving a good balance between accuracy, computational efficiency, and effective uncertainty quantification. With an MSE of 73.41, close to that of Ensemble-DNN, and CP of 92.23\% it outperforms most other models. 
Moreover, the computational efficiency of FBNN (SGHMC-pCN) is highlighted by its speedup factor of 11.98 (Figure \ref{fig:spdup}) over the baseline model BNN-SGHMC.

\subsection{Classification Tasks} 

Next, we evaluate our method based on a set of simulated and real classification problems. The results are summarized in Figures \ref{fig:accuracy} and \ref{fig:spdup}. More details results are provided in Table \ref{tab:performance-classification}.

\begin{figure}[t!]
\centering
\includegraphics[width=1\textwidth]{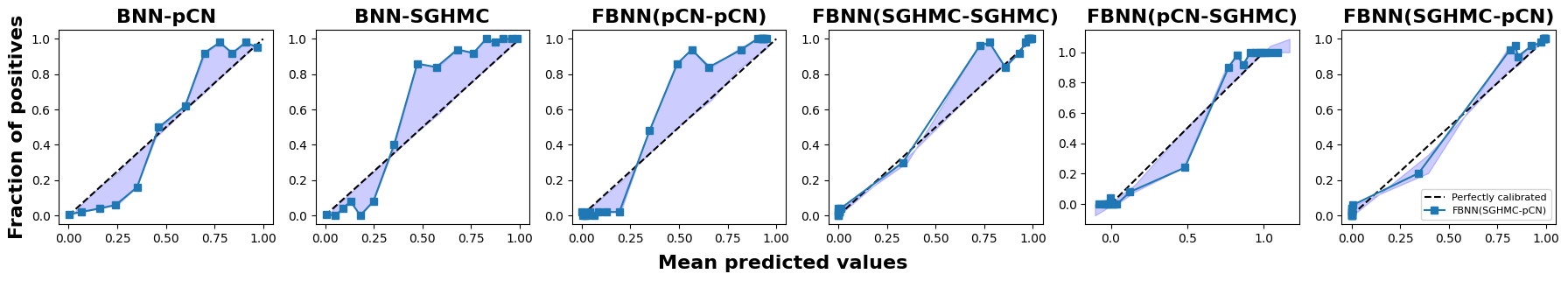}
\caption{Reliability Diagrams for Simulated dataset in classification task. These diagrams incorporate equal frequency binning.}
\label{fig:reliability}
\end{figure}

\paragraph{Simulated Data.}
As before, we start with a simulated dataset with a binary outcome. For this, we use the \verb|make_classification| function from the \verb|sklearn.datasets| package to generate a dataset consisting of 5,000,000 observations and 1000 predictors. 
Accuracy comparison in Figure \ref{fig:accuracy} shows FBNN (SGHMC-pCN), DNN, and Ensemble-DNN exhibit comparable performance and outperform other models. While BNN-RNS-HMC achieves the highest speedup, it significantly underperforms in terms of accuracy. In contrast, FBNN (SGHMC-pCN) provides the second-highest speedup at 32.00, showcasing its computational efficiency relative to BNN-SGHMC. Furthermore, it maintains the second highest accuracy rate of 96\%, indicating an optimal balance between computational efficiency and accuracy.

For this example, BNN-MC-Dropout has the lowest ECE value compared to other methods (see Appendix), but it also have a lower accuracy rate. Among the FBNN variants, FBNN (SGHMC-pCN) presents a low ECE, closely aligning with the ECE value of BNN-MC-Dropout, while providing an accuracy rate similar to DNN.
Moreover, as illustrated in Figure \ref{fig:reliability}, the FBNN variants, particularly FBNN (SGHMC-pCN), appear to be better calibrated across most probability ranges except at the highest probabilities, suggesting a more reliable UQ. Note that in these figures, a model with a reliability curve that closely follows the diagonal line is considered better calibrated, meaning its predicted probabilities are more aligned with actual outcomes.

\paragraph{Adult Data.}

Next, we use the Adult dataset \citep{becker1996}, where the classification task involves predicting whether an individual will earn more or less than \$50,000 per year. 
Figure \ref{fig:accuracy} demonstrates that all the methods achieve comparable accuracy rates, although DNN, Ensemble-DNN, and FBNN(SGHMC-pCN) have the best performance. FBNN (SGHMC-pCN) also stands out as the most computationally efficient method for uncertainty quantification, with a speedup value of 54.91 relative to the baseline BNN approach. A low ECE value for the FBNN (SGHMC-pCN) model signifies its superior performance in terms of uncertainty quantification.

\paragraph{Alzheimer Data.}
We have used the same NACC dataset we previously discussed, but this time as a classification problem. Here, our objective is to predict cognitive status, classifying individuals as either cognitively unimpaired (healthy controls, HC), labeled as class 0, or as having mild cognitive impairment (MCI) due to Alzheimer's disease (AD) or dementia due to AD, labeled as class 1. 

%The BNN-VI and BNN-pCN models show lower accuracies of 72\% and 73\%, respectively, with BNN-pCN also having a high computational demand (2660 seconds). 
Achieving the highest accuracy of 84\% at a relatively low computational cost, FBNN (SGHMC-pCN) surpasses all other models in correctly identifying Alzheimer's disease, making it the most reliable model among those tested. This model not only surpasses the accuracy of the standard DNN and Ensemble-DNN models but also offers a balance between computational efficiency and high accuracy. FBNN (SGHMC-pCN) demonstrates the highest sampling efficiency (speedup of 11), indicating it can achieve high accuracy with a lower computational cost compared to other Bayesian models. Moreover, for the FBNN model implemented on the Alzheimer dataset, the ECE value is low, and the reliability curve closely tracks the diagonal line.

\paragraph{MNIST Dataset.}

The MNIST dataset is commonly used as a benchmark dataset for the hand-written digit classification task \citep{deng2012mnist}. Among the various models evaluated on the MNIST dataset, FBNN (SGHMC-pCN) stands out as the optimal choice, demonstrating exceptional performance across multiple metrics. It achieves the highest accuracy of 94\%, matching the top performance of Ensemble-DNN but with significantly improved efficiency and effectiveness in uncertainty quantification. It exhibits a substantial speedup of 16.77, and the lowest ECE at 0.241, suggesting that it not only provides highly accurate predictions but also reliably estimates the uncertainty associated with these predictions.

\paragraph{CelebA Dataset.}
CelebA \citep{liu2015deep} is an image dataset of celebrity faces annotated with 40 attributes including gender, hair color, age, smiling, etc. The task is to predict hair color, which is either blond $Y = 1$ or non-blond $Y = 0$. 
The FBNN (SGHMC-pCN) model stands out among the alternative methods, showcasing its superiority through several key performance metrics. It achieves the highest accuracy of 85\%, and the highest speedup factor of 61.84, indicating an exceptional balance between computational efficiency and performance. The model achieves the lowest ECE of 0.493, indicating reliability in its predictive uncertainty.

\paragraph{SVHN dataset.}

The Street View House Numbers (SVHN) dataset \citep{netzer2011} includes labelled real-world images of house numbers taken from Google Street View. The images are 32x32 pixels in size and have three color channels (RGB). The goal is to classify digit images into 10 classes. The results demonstrate the superiority of FBNN (SGHMC-pCN) in terms of accuracy and computational efficiency. FBNN (SGHMC-pCN) achieved an accuracy of 96\%, the second highest among all models. The speedup compared to the baseline BNN-SGHMC method was 14.99 times, the highest speedup value recorded. Additionally, the low ECE of 0.203 demonstrates better uncertainty quantification than most other methods, including BNN-VI, BNN-MC-Dropout, and SWAG. 

% Among all the regression and classification datasets, including the performance data for the first 200 SGHMC base samples in Tables \ref{tab:performance-classification} and \ref{tab:performance-regression} reveals that these initial samples do not exhibit high quality when evaluated in terms of MSE for regression and accuracy for classification. The comparison indicates that, especially in regression cases, FBNN (SGHMC-pCN) significantly outperforms the 200 base SGHMC samples in terms of spdup across all cases examined. Furthermore, in classification datasets, FBNN (SGHMC-pCN) achieves significantly better results in terms of ECE compared to the 200 base SGHMC samples.

\section{Conclusion}
In this paper, we have proposed an innovative CES framework called FBNN, specifically designed to enhance the computational efficiency and scalability of BNN for high-dimensional data. Our primary goal is to provide a robust solution for uncertainty quantification in high-dimensional spaces, leveraging the power of Bayesian inference while mitigating the computational bottlenecks traditionally associated with BNN.

In our numerical experiments, we have successfully applied several variants of FBNN, including different configurations with BNN, to regression and classification tasks on both synthetic and real datasets. Remarkably, the FBNN variant incorporating SGHMC for calibration and pCN for sampling, denoted as FBNN (SGHMC-pCN), not only matches the predictive accuracy of traditional BNN but also offers substantial computational advantages. 
More specifically, our numerical experiments across various regression and classification tasks consistently demonstrate the superiority of the FBNN (SGHMC-pCN) method over traditional BNNs and DNNs. In regression tasks, FBNN (SGHMC-pCN) demonstrates a competitive MSE while significantly enhancing computational efficiency, achieving notable speedups compared to baseline models. This efficiency does not come at the expense of accuracy, as evidenced by the competitive MSE values and robust uncertainty quantification metrics. In classification tasks, FBNN (SGHMC-pCN) stands out by achieving high accuracy rates and low ECE values, which indicate reliable uncertainty quantification.

The superior performance of FBNN (SGHMC-pCN) can be attributed to the complementary strengths of SGHMC and pCN. SGHMC excels at broad exploration of the parameter space, providing an effective means for understanding the global structure during the calibration step. On the other hand, pCN is adept at efficient sampling around modes, offering a valuable tool for capturing local intricacies in the distribution during the final sampling step. By combining these samplers within the FBNN framework, we achieve a balanced approach between exploration (calibration with SGHMC) and exploitation (final sampling with pCN). %This adaptability is particularly beneficial in high-dimensional and complex Bayesian Neural Network (BNN) settings, where SGHMC's ability to explore broadly complements pCN's efficiency in mode sampling, resulting in a robust and effective optimization strategy.

Future work could involve extending our method to more complex problems (e.g., spatiotemporal data) and complex network structures (e.g., graph neural networks). Additionally, future research could focus on improving the emulation step by optimizing the DNN architecture. Finally, our method could be further improved by embedding the sampling algorithm in an adaptive framework similar to the method of \citet{Zhang2018}.

\section*{Acknowledgements}
The authors thank the Editor, Action Editor, and anonymous reviewers for their insightful suggestions and constructive feedback, which significantly improved the article. This work was supported by NSF grants NCS-FR-2319618 and DMS-2134256.

%\clearpage

\bibliography{arxiv}
\bibliographystyle{tmlr}

\newpage

\appendix

\setcounter{table}{0}
\renewcommand{\thetable}{A\arabic{table}}

\section*{Appendix A: Comparing Various Deep Learning Techniques for Regression and Classification Problems}

\begin{table*}[h!]
\centering
\caption{Performance of various deep learning methods based on regression problems. For ESS, minimum, median, and maximum values are provided in parenthesis.}
\fontsize{7pt}{8pt}\selectfont
  \begin{tabular}{lcccccccc}
    \textbf{Dataset} & \textbf{Method} & \textbf{MSE} & \textbf{CP} & \textbf{Time (s)} & \textbf{ESS} & \textbf{minESS/s} & \textbf{spdup} \\
    \hline
Simulated & DNN & 0.71 & - & 3531 & - & - & - \\
    Dataset    & Ensemble-DNN & 0.74 & 91.3\% & 14633 & - & - & - \\
    & BNN-VI & 1.14 & 88.8\% & 8408 & - & - & -\\
    & BNN-Lasso & 0.98 & 87.4\% & 6941 & - & - & -\\
    & BNN-MC-Dropout & 0.83 & 93.4\% & 2861 & - & - & -\\
    & BNN-SGHMC & 0.79 & 91.6\% & 41281 & (109, 829, 1642) & 0.002 & 1.00 \\
    & BNN-SGHMC(first 200) & 4.48 & 92.8\% & 3970 & (21, 66, 162)  & 0.005 &  2.01\\    
    & SWAG & 0.73 & 90.0\% & 13488 & (93, 1163, 1542)  & 0.006 & 2.61\\
    & BNN-RNS-HMC & 1.62 & 87.5\% & 6046 & (126, 948, 1510) & 0.021 & 7.89\\
    & BNN-pCN & 0.81 & 89.3\% & 42523 & (107, 844, 1533) & 0.002 &  0.95\\        
    & FBNN (pCN-SGHMC) & 0.96 & 79.6\% & 4512 & (132, 1241, 1615)  & 0.029 & 11.08 \\      
    & FBNN (pCN-pCN) & 0.86 & 85.2\% & 4631 & (94, 975, 1621)  & 0.020 & 7.69 \\
    & FBNN (SGHMC-SGHMC) & 0.82 & 90.1\% & 4497 & (137, 1236, 1638)  & 0.030 & 11.53\\
    & FBNN (SGHMC-pCN) & 0.74 & 92.2\% & 4312 &  (186, 912, 1606)  & 0.043 & 16.33 \\    
    \hline
    Wine & DNN & 0.43 & -  & 26 & - & - & - \\
    Quality & Ensemble-DNN & 0.41 & 47.4\% & 137 & - & - & - \\
    & BNN-VI & 0.66 & 39.5\% & 28 & - & - & -\\
    & BNN-Lasso & 0.63 & 40.9\% & 42 & - & - & -\\
    & BNN-MC-Dropout & 0.67 & 32.3\% & 12 & - & - & -\\    
    & BNN-SGHMC & 0.53 & 51.3\% & 505 & (111, 837, 1538) & 0.219 & 1.00 \\
    & BNN-SGHMC(first 200) & 2.75 & 54.6\% & 61 & (13, 111, 150) & 0.213 & 0.91 \\  
    & SWAG & 0.53 & 48.2\% & 97 & (98, 1021, 1489) & 1.010 & 4.39\\
    & BNN-RNS-HMC & 0.62 & 44.7\% & 38 & (107, 925, 1520) & 2.815 & 12.24\\
    & BNN-pCN & 0.65 & 51.1\% & 620 & (99, 1003, 1532) & 0.159 & 0.69 \\
    & FBNN (pCN-SGHMC) & 0.52 & 32.2\% & 68 & (91, 912, 1533) & 1.338 & 5.95 \\   
    & FBNN (pCN-pCN) & 0.65 & 24.5\% & 67 & (105, 1087, 1540) & 1.567 & 6.86 \\  
    & FBNN (SGHMC-SGHMC) & 0.50 & 40.0\% & 70 & (77, 806, 1536) & 1.100 & 4.78 \\
    & FBNN (SGHMC-pCN) & 0.52 & 48.1\% & 57 & (92, 897, 1536) & 1.614 & 7.33 \\    
    \hline
    Boston & DNN & 3.21 & -  & 14 & - & - & - \\
    Housing & Ensemble-DNN & 3.17 & 72.1\% & 74 & - & - & - \\
    & BNN-VI & 7.60 & 83.7\% & 85 & - & - & -\\
    & BNN-Lasso & 6.20 & 79.2\% & 68 & - & - & -\\
    & BNN-MC-Dropout & 10.12 & 81.2\% & 91 & - & - & -\\    
    & BNN-SGHMC & 3.83 & 75.3\% & 888 & (76, 649, 1536) & 0.085 & 1.00 \\
    & BNN-SGHMC(first 200) & 7.40 & 66.9\% & 86 & (9, 87, 150) & 0.104 & 1.22 \\ 
    & SWAG & 5.81 & 71.9\% & 104 & (68, 724, 1532) & 0.653 & 7.22\\
    & BNN-RNS-HMC & 9.42 & 73.4\% & 76 & (73, 1032, 1504) & 0.960 & 10.6\\
    & BNN-pCN & 3.25 & 79.3\% & 901 & (76, 649, 1536) & 0.084 & 0.88 \\       
    & FBNN (pCN-SGHMC) & 4.16 & 41.7\% & 186 & (71, 965, 1543) & 0.381 & 4.22 \\
    & FBNN (pCN-pCN) & 3.81 & 47.1\% & 186 & (80, 966, 1541) & 0.430 & 4.78 \\ 
    & FBNN (SGHMC-SGHMC) & 4.15 & 48.9\% & 94 & (69, 979, 1542) & 0.734 & 8.22 \\
    & FBNN (SGHMC-pCN) & 3.82 & 81.1\% & 91 & (93, 938, 1543) & 1.021 & 11.94 \\
    \hline
    Alzheimer & DNN & 0.49 & - & 326 & - & - & - \\
    Dataset & Ensemble-DNN & 0.42 & 89.3\% & 1597 & - & - & -\\
    & BNN-VI & 0.53 & 87.6\% & 341 & - & - & -\\
    & BNN-Lasso & 0.52 & 83.5\% & 561 & - & - & -\\
    & BNN-MC-Dropout & 0.60 & 92.8\% & 268 & - & - & -\\    
    & BNN-SGHMC & 0.49 & 91.6\% & 6524 & (102, 973, 1376) & 0.015 & 1.00 \\
    & BNN-SGHMC(first 200) & 3.17 & 72.7\% & 641 & (7, 82, 150) & 0.011 & 0.69 \\ 
    & SWAG & 0.74 & 89.3\% & 1214 & (106, 1002, 1542) & 0.087 & 5.58\\
    & BNN-RNS-HMC & 0.61 & 92.4\% & 7324 & (96, 892, 1531) & 0.013 & 0.83\\
    & BNN-pCN & 0.51 & 89.9\% & 6212 & (120, 1092, 1448) & 0.019 & 1.23 \\       
    & FBNN (pCN-SGHMC) & 0.50 & 90.2\% & 643 & (116, 994, 1504) & 0.180 & 11.53 \\
    & FBNN (pCN-pCN) & 0.56 & 91.4\% & 682 & (108, 998, 1498) & 0.171 & 10.93 \\ 
    & FBNN (SGHMC-SGHMC) & 0.55 & 88.4\% & 671 & (118, 1012, 1541) & 0.176 & 11.24 \\
    & FBNN (SGHMC-pCN) & 0.48 & 91.6\% & 632 & (149, 984, 1497) & 0.218 & 13.97 \\
    \hline
    YearPredictionMSD & DNN & 74.54 & - & 1569 & - & - & - \\
    Dataset & Ensemble-DNN & 72.89 & 90.47\% & 7929 & - & - & -\\
    & BNN-VI & 78.21 & 88.43\% & 2146 & - & - & -\\
    & BNN-Lasso & 79.44 & 89.04\% & 3243 & - & - & -\\
    & BNN-MC-Dropout & 83.08 & 87.89\% & 1287 & - & - & -\\    
    & BNN-SGHMC & 81.67 & 83.81\% & 25533 & (122, 1005, 1540) & 0.004 & 1.00 \\
    & BNN-SGHMC(first 200) & 105.92 & 94.13\% & 2613 & (7, 84, 149) & 0.002 & 0.56 \\ 
    & SWAG & 93.87 & 86.33\% & 3841 & (113, 987, 1537) & 0.029 & 7.35 \\
    & BNN-RNS-HMC & 92.82 & 80.19\% & 17289 & (109, 925, 1563) & 0.006 & 1.57 \\
    & BNN-pCN & 82.45 & 85.74\% & 25655 & (124, 873, 1631) & 0.004 & 1.20 \\       
    & FBNN (pCN-SGHMC) & 74.92 & 88.73\% & 2815 & (143, 1049, 1676) & 0.050 & 10.63\\
    & FBNN (pCN-pCN) & 74.03 & 90.45\% & 3046 & (138, 992, 1618) & 0.045 & 9.48\\ 
    & FBNN (SGHMC-SGHMC) & 76.69 & 90.40\% & 2974 & (146, 973, 1599) & 0.049 & 10.27\\
    & FBNN (SGHMC-pCN) & 73.41 & 92.23\% & 2724 & (156, 934, 1608) & 0.057 & 11.98\\
  \end{tabular}
\label{tab:performance-regression} 
\end{table*}

\newpage

\begin{table*}[h!]
\caption{Performance of various deep learning methods based on classification problems.} 
\fontsize{7pt}{7pt}\selectfont
  \centering
  \begin{tabular}{lccccccc}
    \textbf{Dataset} & \textbf{Method} & \textbf{Acc} & \textbf{Time(s)} & \textbf{ESS(min,med,max)} & \textbf{minESS/s} & \textbf{spdup}  &  \textbf{ECE} \\
    \hline
    Simulated  & DNN & 96\% & 4257 & - & - & - & -\\
    Dataset & Ensemble-DNN & 97\% & 17415 & - & - & - & 0.382\\
    & BNN-VI & 90\% & 4275 & - & - & - & 0.399\\
    & BNN-Lasso & 92\% & 3189 & - & - & - & 0.363\\
    & BNN-MC-Dropout & 87\% & 2912 & - & - & - & 0.277\\
    & BNN-SGHMC & 95\% & 43841 & (47, 212, 1459) & 0.001 & 1.00 & 0.471\\
    & BNN-SGHMC(first 200) & 83\% & 4218 & (21, 59, 156) & 0.004 & 4.64 & 0.498\\
    & SWAG & 81\% & 4731 & (81, 773, 1368) & 0.017 & 15.97 & 0.482\\
    & BNN-RNS-HMC & 69\% & 1309 & (135, 1190, 1493) & 0.103 & 96.20 & 0.513\\
    & BNN-pCN & 91\% & 49774 & (36, 207, 1417) & 0.001 & 0.67 & 0.475\\
    & FBNN (pCN-SGHMC) & 93\% & 5179 & (134, 959, 1419) & 0.051 & 24.13 & 0.409\\
    & FBNN (pCN-pCN) & 91\% & 4858 & (146, 921, 1412) & 0.058 & 28.03 & 0.423\\
    & FBNN (SGHMC-SGHMC) & 95\% & 4517 & (149, 891, 1540) & 0.032 & 30.76 & 0.406\\
    & FBNN (SGHMC-pCN) & 96\% & 4489 & (154, 911, 1602) & 0.070 & 32.00 & 0.396\\

    \hline
    Adult  & DNN & 85\% & 426 & - & - & - & -\\
    Dataset & Ensemble-DNN & 84\% & 2153 & - & - & - & 0.556\\
    & BNN-VI & 80\% & 562 & - & - & - & 0.642\\
    & BNN-Lasso & 83\% & 256 & - & - & - & 0.631\\
    & BNN-MC-Dropout & 82\% & 187 & - & - & - & 0.540\\
    & BNN-SGHMC & 83\% & 5979 & (16, 202, 1520) & 0.002 & 1.00 & 0.574\\  
    & BNN-SGHMC(first 200) & 78\% & 581 & (1, 41, 148) & 0.002 & 0.95 & 0.594\\
    & SWAG & 79\% & 1641 & (47, 912, 1532) & 0.028 & 10.70 & 0.668\\
    & BNN-RNS-HMC & 72\% & 6110 & (89, 960, 1530) & 0.014 & 5.44 & 0.658\\
    & BNN-pCN & 83\% & 6227 & (9, 117, 1518) & 0.001 & 0.54 & 0.616\\
    & FBNN (pCN-SGHMC) & 82\% & 642 & (87, 892, 1539) & 0.135 & 50.63 & 0.580\\  
    & FBNN (pCN-pCN) & 82\% & 639 & (88, 890, 1540) & 0.137 & 51.46 & 0.592\\
    & FBNN (SGHMC-SGHMC) & 83\% & 612 & (68, 941, 1541) & 0.111 & 41.52 & 0.583\\
    & FBNN (SGHMC-pCN) & 84\% & 609 & (89, 875, 1539) & 0.146 & 54.91 & 0.576\\
    \hline
    Alzheimer  & DNN & 82\% & 51 & - & - & - & - \\
    Dataset & Ensemble-DNN & 83\% & 262 & - & - & - & 0.542\\
    & BNN-VI & 72\% & 61 & - & - & - & 0.546\\
    & BNN-Lasso & 76\% & 256 & - & - & - & 0.524\\
    & BNN-MC-Dropout & 76\% & 12 & - & - & - & 0.429\\
    & BNN-SGHMC & 81\% & 2736 & (81, 588, 1526) & 0.029 & 1.00 & 0.499\\  
    & BNN-SGHMC(first 200) & 69\% & 282 & (8, 84, 149) & 0.028 & 0.96 & 0.523\\
    & SWAG & 81\% & 312 & (72, 913, 1562) & 0.231 & 7.69 & 0.508\\
    & BNN-RNS-HMC & 58\% & 293 & (84, 915, 1540) & 0.286 & 9.55 & 0.521\\
    & BNN-pCN & 73\% & 2660 & (71, 424, 1534) & 0.026 & 0.90 & 0.469\\
    & FBNN (pCN-SGHMC) & 76\% & 277 & (76, 947, 1542) & 0.274 & 9.26 & 0.568\\  
    & FBNN (pCN-pCN) & 77\% & 274 & (70, 931, 1542) & 0.255 & 8.33 & 0.377\\
    & FBNN (SGHMC-SGHMC) & 80\% & 278 & (81, 973, 1538) & 0.291 & 8.63 & 0.448\\
    & FBNN (SGHMC-pCN) & 84\% & 280 & (92, 914, 1535) & 0.328 & 11.09 & 0.376\\
    \hline
    Mnist  & DNN & 92\% & 231 & - & - & - & - \\
    Dataset & Ensemble-DNN & 94\% & 1129 & - & - & - & 0.312\\
    & BNN-VI & 86\% & 273 & - & - & - & 0.417\\
    & BNN-Lasso & 87\% & 184 & - & - & - & 0.445\\
    & BNN-MC-Dropout & 89\% & 212 & - & - & - & 0.328\\
    & BNN-SGHMC & 90\% & 8641 & (15, 364, 1456) & 0.001 & 1.00 & 0.280\\  
    & BNN-SGHMC(first 200) & 78\% & 916 & (1, 34, 149) & 0.001 & 0.62 & 0.271\\
    & SWAG & 83\% & 1294 & (15, 431, 1376) & 0.011 & 7.72 & 0.327\\
    & BNN-RNS-HMC & 69\% & 4541 & (14, 372, 1394) & 0.003 & 2.05 & 0.349\\
    & BNN-pCN & 88\% & 8912 & (17, 398, 1471) & 0.001 & 1.26 & 0.321\\
    & FBNN (pCN-SGHMC) & 88\% & 927 & (15, 412, 1383) & 0.016 & 10.78 & 0.352\\  
    & FBNN (pCN-pCN) & 90\% & 931 & (16, 393, 1421) & 0.017 & 11.45 & 0.312\\
    & FBNN (SGHMC-SGHMC) & 91\% & 923 & (18, 409, 1461) & 0.019 & 13.01 & 0.283\\
    & FBNN (SGHMC-pCN) & 94\% & 914 & (23, 474, 1521) & 0.025 & 16.77 & 0.241\\
    \hline
    celebA  & DNN & 80\% & 3689 & - & - & - & - \\
    Dataset & Ensemble-DNN & 82\% & 15445 & - & - & - & 0.569\\
    & BNN-VI & 80\% & 1132 & - & - & - & 0.622\\
    & BNN-Lasso & 79\% & 2159 & - & - & - & 0.561\\
    & BNN-MC-Dropout & 78\% & 1641 & - & - & - & 0.512\\
    & BNN-SGHMC & 81\% & 17234 & (19, 383, 1537) & 0.001 & 1.00 & 0.567 \\  
    & BNN-SGHMC(first 200) & 69\% & 1849 & (1, 85, 149) & 0.001 & 0.63 & 0.642 \\
    & SWAG & 70\% & 8913 & (85, 1014, 1467) & 0.009 & 8.65 & 0.534\\
    & BNN-RNS-HMC & 72\% & 1835 & (72, 951, 1494) & 0.039 & 35.59 & 0.612\\
    & BNN-pCN & 76\% & 19676 & (19, 331, 1538) & 0.001 & 0.95 & 0.529\\
    & FBNN (pCN-SGHMC) & 80\% & 1972 & (99, 1155, 1542) & 0.050 & 45.53 & 0.565\\  
    & FBNN (pCN-pCN) & 79\% & 1951 & (116, 1155, 1542) & 0.059 & 53.93 & 0.542\\
    & FBNN (SGHMC-SGHMC) & 81\% & 1904 & (93, 978, 1544) & 0.048 & 44.30 & 0.568\\
    & FBNN (SGHMC-pCN) & 85\% & 1892 & (129, 785, 1517) & 0.068 & 61.84 & 0.493\\
    \hline
    SVHN  & DNN & 96\% & 3748 & - & - & - & - \\
    Dataset & Ensemble-DNN & 97\% & 17415 & - & - & - & 0.382\\
    & BNN-VI & 90\% & 4275 & - & - & - & 0.399\\
    & BNN-Lasso & 92\% & 3189 & - & - & - & 0.363\\
    & BNN-MC-Dropout & 87\% & 2912 & - & - & - & 0.277\\
    & BNN-SGHMC & 91\% & 18639 & (31, 379, 1528) & 0.001 & 1.00 & 0.221\\  
    & BNN-SGHMC(first 200) & 74\% & 2515 & (5, 67, 262) & 0.001 & 1.19 & 0.246\\
    & SWAG & 83\% & 6294 & (78, 499, 1383) & 0.012 & 7.45 & 0.282\\
    & BNN-RNS-HMC & 71\% & 9083 & (19, 410, 1384) & 0.002 & 1.25 & 0.299\\
    & BNN-pCN & 93\% & 17812 & (27, 398, 1317) & 0.001 & 0.91 & 0.275\\
    & FBNN (pCN-SGHMC) & 93\% & 1931 & (36, 391, 1403) & 0.018 & 11.21 & 0.263\\  
    & FBNN (pCN-pCN) & 91\% & 1927 & (29, 372, 1280) & 0.015 & 9.04 & 0.318\\
    & FBNN (SGHMC-SGHMC) & 96\% & 1891 & (34, 412, 1464) & 0.017 & 10.81 & 0.248\\
    & FBNN (SGHMC-pCN) & 96\% & 1884 & (47, 486, 1533) & 0.024 & 14.99 & 0.223\\

  \end{tabular} 
  \label{tab:performance-classification}
\end{table*}

\end{document}